\documentclass[letterpaper,twocolumn,10pt]{article}

\newif\ifarxiv
\arxivtrue

\usepackage{usenix}

\usepackage{amsmath,amssymb,amsthm}
\usepackage{booktabs}
\usepackage{graphicx}

\usepackage{pifont}                      
\usepackage{tikz}
\usetikzlibrary{arrows.meta, positioning, fit, backgrounds, calc, shapes.geometric, shadows}

\definecolor{mosKey}{HTML}{1F4E79}\definecolor{mosKeyBg}{HTML}{E7EFF7}
\definecolor{mosVer}{HTML}{1B6B4C}\definecolor{mosVerBg}{HTML}{E4F1EA}
\definecolor{mosAdv}{HTML}{9B2226}\definecolor{mosAdvBg}{HTML}{FAEAEA}
\definecolor{mosNeu}{HTML}{55585C}\definecolor{mosNeuBg}{HTML}{F1F2F3}
\definecolor{mosAcc}{HTML}{A96A00}\definecolor{mosAccBg}{HTML}{FBF0DC}
\tikzset{
  mosfig/.style={>={Latex[length=3.4pt,width=2.8pt]}, font=\scriptsize, line width=0.45pt,
                 execute at begin picture={\hyphenpenalty=10000\exhyphenpenalty=10000\relax}},
  card/.style 2 args={draw=#1!45, fill=#1!7, rounded corners=2.2pt, align=center, inner sep=3.6pt,
                      text width=#2, minimum height=7mm},
  plain/.style={draw=mosNeu!45, fill=mosNeu!6, rounded corners=2.2pt, align=center, inner sep=3.6pt},
  ghost/.style={draw=mosNeu!45, fill=white, densely dashed, text=mosNeu!85, rounded corners=2.2pt,
                align=center, inner sep=3.6pt},
  flow/.style   ={->, draw=mosNeu!85, line width=0.65pt},
  keyflow/.style={->, draw=mosKey, line width=0.65pt, densely dashed},
  advflow/.style={->, draw=mosAdv, line width=0.6pt, dotted},
  lab/.style ={font=\tiny, inner sep=1.4pt, text=mosNeu},
  klab/.style={lab, text=mosKey},
  alab/.style={lab, text=mosAdv},
  elab/.style={lab, fill=white, inner sep=1.6pt, rounded corners=1pt},
  zone/.style={rounded corners=3pt, line width=0.4pt, inner sep=4.4pt},
  scope/.style={zone, draw=mosKey!40, densely dashed},
  chan/.style ={zone, draw=mosNeu!35, densely dotted, fill=mosNeu!4},
}
\newcommand{\cardn}[7]{%
  \node[draw=#2!55, shade, top color=white, bottom color=#2!13, line width=0.5pt,
        rounded corners=2.4pt, align=center,
        inner sep=3pt, text width=#3, minimum height=#5,
        drop shadow={opacity=0.10, shadow xshift=0.45pt, shadow yshift=-0.45pt}] (#1) at #4
       {\rule{0pt}{3.9mm}\\[-1.3mm]#7};%
  \begin{scope}
    \clip[rounded corners=2.4pt] (#1.south west) rectangle (#1.north east);
    \shade[top color=#2!92, bottom color=#2] (#1.north west) rectangle ([yshift=-4.1mm]#1.north east);
  \end{scope}
  \node[anchor=north, inner sep=0pt, yshift=-1.15mm, text=white, font=\scriptsize\bfseries]
       at (#1.north) {#6};%
}
\tikzset{
  keyicon/.pic={\draw[line width=0.4pt, fill=mosKeyBg] (0,0) circle (1.15pt);
    \draw[line width=0.5pt] (1.15pt,0) -- (4.1pt,0);
    \draw[line width=0.5pt] (2.9pt,0) -- (2.9pt,-1.35pt);
    \draw[line width=0.5pt] (3.9pt,0) -- (3.9pt,-1.35pt);},
  person/.pic={\draw[line width=0.5pt] (0,2.1pt) circle (1.5pt);
    \draw[line width=0.5pt] (-2.2pt,-2.3pt) .. controls (-2.2pt,0.7pt) and (2.2pt,0.7pt) .. (2.2pt,-2.3pt);},
  eye/.pic={\draw[line width=0.5pt] (-3.1pt,0) .. controls (-1.5pt,2.3pt) and (1.5pt,2.3pt) .. (3.1pt,0)
      .. controls (1.5pt,-2.3pt) and (-1.5pt,-2.3pt) .. (-3.1pt,0);
    \draw[line width=0.5pt, fill] (0,0) circle (0.8pt);},
  cam/.pic={\draw[line width=0.5pt, rounded corners=0.6pt] (-3pt,-2.1pt) rectangle (3pt,1.7pt);
    \draw[line width=0.5pt] (-1.2pt,1.7pt) -- (-0.7pt,2.7pt) -- (0.7pt,2.7pt) -- (1.2pt,1.7pt);
    \draw[line width=0.5pt] (0,-0.2pt) circle (1.2pt);},
  skelA/.pic={\draw[line width=0.42pt] (0,7.4pt) circle (1.7pt);
    \draw[line width=0.42pt] (0,5.7pt) -- (0,0.6pt);
    \draw[line width=0.42pt] (-2.9pt,2.0pt) -- (0,4.4pt) -- (2.7pt,2.9pt);
    \draw[line width=0.42pt] (-2.4pt,-3.8pt) -- (0,0.6pt) -- (2.2pt,-3.8pt);},
  skelB/.pic={\draw[line width=0.42pt] (0,7.4pt) circle (1.7pt);
    \draw[line width=0.42pt] (0,5.7pt) -- (0.4pt,0.6pt);
    \draw[line width=0.42pt] (-2.5pt,5.6pt) -- (0,4.4pt) -- (2.6pt,5.4pt);
    \draw[line width=0.42pt] (-1.2pt,-3.8pt) -- (0.4pt,0.6pt) -- (2.6pt,-3.6pt);},
  skelC/.pic={\draw[line width=0.42pt] (0,7.4pt) circle (1.7pt);
    \draw[line width=0.42pt] (0,5.7pt) -- (-0.4pt,0.6pt);
    \draw[line width=0.42pt] (-2.7pt,3.4pt) -- (0,4.4pt) -- (2.9pt,1.9pt);
    \draw[line width=0.42pt] (-2.8pt,-3.5pt) -- (-0.4pt,0.6pt) -- (1.4pt,-3.9pt);},
  rep3d/.pic={
    \fill[mosNeu!16] (0,-0.33) ellipse (0.17 and 0.037);
    \draw[line width=0.5pt, mosNeu!85, line cap=round]
      (0,0.20) -- (0,0.00) (-0.14,0.05) -- (0,0.16) -- (0.13,0.07)
      (-0.11,-0.29) -- (0,0.00) -- (0.11,-0.29);
    \draw[line width=0.5pt, mosNeu!85] (0,0.255) circle (0.05);},
  reppix/.pic={
    \draw[line width=1.95pt, mosAcc!70, line cap=round] (0,0.20) -- (0,0.00);
    \draw[line width=1.25pt, mosAcc!70, line cap=round] (-0.14,0.05) -- (0,0.16) -- (0.13,0.07);
    \draw[line width=1.45pt, mosAcc!70, line cap=round] (-0.11,-0.29) -- (0,0.00) -- (0.11,-0.29);
    \fill[mosAcc!70] (0,0.255) circle (0.062);},
  rep2d/.pic={
    \draw[line width=1.95pt, mosAcc!20, line cap=round] (0,0.20) -- (0,0.00);
    \draw[line width=1.25pt, mosAcc!20, line cap=round] (-0.14,0.05) -- (0,0.16) -- (0.13,0.07);
    \draw[line width=1.45pt, mosAcc!20, line cap=round] (-0.11,-0.29) -- (0,0.00) -- (0.11,-0.29);
    \fill[mosAcc!20] (0,0.255) circle (0.062);
    \draw[mosKey!55, line width=0.3pt]
      (0,0.255) -- (0,0.16) -- (0,0.00) (-0.14,0.05) -- (0,0.16) -- (0.13,0.07)
      (-0.11,-0.29) -- (0,0.00) -- (0.11,-0.29);
    \foreach \pp in {(0,0.255),(0,0.16),(0,0.00),(-0.14,0.05),(0.13,0.07),(-0.11,-0.29),(0.11,-0.29)}{%
      \fill[mosKey] \pp circle (0.027);}},
  rep3dr/.pic={
    \draw[line width=0.5pt, mosNeu!30, line cap=round]
      (0,0.20) -- (0,0.00) (-0.14,0.05) -- (0,0.16) -- (0.13,0.07)
      (-0.11,-0.29) -- (0,0.00) -- (0.11,-0.29);
    \draw[line width=0.5pt, mosNeu!30] (0,0.255) circle (0.05);
    \fill[mosNeu!16] (0.03,-0.33) ellipse (0.17 and 0.037);
    \draw[line width=0.5pt, mosKey, line cap=round]
      (0.03,0.20) -- (0.02,0.00) (-0.12,0.03) -- (0.03,0.16) -- (0.15,0.09)
      (-0.09,-0.29) -- (0.02,0.00) -- (0.13,-0.28);
    \draw[line width=0.5pt, mosKey] (0.03,0.255) circle (0.05);},
  skel/.pic={
    \draw[line width=0.42pt] (0,7.4pt) circle (1.7pt);
    \draw[line width=0.42pt] (0,5.7pt) -- (0,0.6pt);
    \draw[line width=0.42pt] (-2.9pt,2.0pt) -- (0,4.4pt) -- (2.7pt,2.9pt);
    \draw[line width=0.42pt] (-2.4pt,-3.8pt) -- (0,0.6pt) -- (2.2pt,-3.8pt);},
  link/.pic={
    \draw[line width=0.45pt] (-0.6pt,0) arc (0:180:1.5pt and 2pt);
    \draw[line width=0.45pt] (0.6pt,0) arc (180:360:1.5pt and 2pt);},
  forge/.pic={
    \draw[line width=0.45pt] (-3pt,-1pt) .. controls (-1pt,2.4pt) and (0pt,-2.4pt) .. (3pt,1pt);},
  anon/.pic={\fill (0,1.95pt) circle (1.35pt);
    \fill (-2.6pt,-2.4pt) .. controls (-2.6pt,0.7pt) and (-1.35pt,1.15pt) .. (0,1.15pt)
          .. controls (1.35pt,1.15pt) and (2.6pt,0.7pt) .. (2.6pt,-2.4pt) -- cycle;},
  check/.pic={\draw[line width=0.75pt, line cap=round, line join=round]
      (-2.5pt,0.3pt) -- (-0.8pt,-1.5pt) -- (2.7pt,2.2pt);},
  screen/.pic={\draw[line width=0.45pt, rounded corners=0.6pt] (-3.2pt,-1.1pt) rectangle (3.2pt,2.7pt);
    \draw[line width=0.45pt] (0,-1.1pt) -- (0,-2.4pt);
    \draw[line width=0.45pt] (-1.7pt,-2.4pt) -- (1.7pt,-2.4pt);},
}
\newcommand{\bits}[7]{%
  \begin{scope}[shift={(#2,#3)}]
    \foreach \b [count=\i from 0] in {#7}{%
      \pgfmathsetmacro{\tint}{\b > 0.5 ? 82 : 14}%
      \fill[#1!\tint] (\i*#4,0) rectangle (\i*#4+#4,#5);}%
    \draw[#1!60, line width=0.45pt] (0,0) rectangle (#6*#4,#5);
  \end{scope}}

\tikzset{poseA/.pic={%
  \draw[line width=0.42pt, line cap=round, line join=round] (0.00pt,0.00pt) -- (-0.40pt,-0.64pt) -- (-0.81pt,-3.31pt) -- (-0.74pt,-6.24pt) -- (-1.20pt,-6.53pt);
  \draw[line width=0.42pt, line cap=round, line join=round] (0.00pt,0.00pt) -- (0.41pt,-0.58pt) -- (0.84pt,-3.28pt) -- (0.78pt,-6.29pt) -- (0.96pt,-6.60pt);
  \draw[line width=0.42pt, line cap=round, line join=round] (0.00pt,0.00pt) -- (0.03pt,0.88pt) -- (0.05pt,1.86pt) -- (0.03pt,2.25pt) -- (-0.09pt,3.77pt) -- (0.03pt,4.35pt);
  \draw[line width=0.42pt, line cap=round, line join=round] (0.03pt,2.25pt) -- (-0.53pt,3.07pt) -- (-1.39pt,3.11pt) -- (-1.85pt,1.42pt) -- (-2.54pt,-0.12pt);
  \draw[line width=0.42pt, line cap=round, line join=round] (0.03pt,2.25pt) -- (0.49pt,3.09pt) -- (1.40pt,3.16pt) -- (1.80pt,1.47pt) -- (2.18pt,-0.17pt);
  \draw[line width=0.42pt] (0.03pt,4.35pt) circle (0.94pt);}}
\tikzset{poseB/.pic={%
  \draw[line width=0.42pt, line cap=round, line join=round] (0.00pt,0.00pt) -- (-0.40pt,-0.64pt) -- (-0.82pt,-3.31pt) -- (-0.76pt,-6.25pt) -- (-1.23pt,-6.53pt);
  \draw[line width=0.42pt, line cap=round, line join=round] (0.00pt,0.00pt) -- (0.41pt,-0.58pt) -- (0.84pt,-3.28pt) -- (0.78pt,-6.29pt) -- (0.95pt,-6.61pt);
  \draw[line width=0.42pt, line cap=round, line join=round] (0.00pt,0.00pt) -- (0.05pt,0.88pt) -- (0.08pt,1.86pt) -- (0.04pt,2.26pt) -- (-0.06pt,3.77pt) -- (-0.02pt,4.39pt);
  \draw[line width=0.42pt, line cap=round, line join=round] (0.04pt,2.26pt) -- (-0.51pt,3.07pt) -- (-1.36pt,3.14pt) -- (-2.29pt,1.56pt) -- (-2.74pt,2.16pt);
  \draw[line width=0.42pt, line cap=round, line join=round] (0.04pt,2.26pt) -- (0.51pt,3.09pt) -- (1.43pt,3.14pt) -- (1.81pt,1.44pt) -- (2.09pt,-0.21pt);
  \draw[line width=0.42pt] (-0.02pt,4.39pt) circle (0.94pt);}}
\tikzset{poseC/.pic={%
  \draw[line width=0.42pt, line cap=round, line join=round] (0.00pt,0.00pt) -- (-0.40pt,-0.64pt) -- (-0.81pt,-3.31pt) -- (-0.77pt,-6.25pt) -- (-1.24pt,-6.54pt);
  \draw[line width=0.42pt, line cap=round, line join=round] (0.00pt,0.00pt) -- (0.41pt,-0.58pt) -- (0.84pt,-3.28pt) -- (0.79pt,-6.29pt) -- (0.95pt,-6.62pt);
  \draw[line width=0.42pt, line cap=round, line join=round] (0.00pt,0.00pt) -- (0.03pt,0.88pt) -- (0.06pt,1.87pt) -- (0.03pt,2.26pt) -- (-0.10pt,3.77pt) -- (-0.06pt,4.41pt);
  \draw[line width=0.42pt, line cap=round, line join=round] (0.03pt,2.26pt) -- (-0.54pt,3.06pt) -- (-1.39pt,3.13pt) -- (-2.41pt,1.91pt) -- (-3.84pt,2.78pt);
  \draw[line width=0.42pt, line cap=round, line join=round] (0.03pt,2.26pt) -- (0.48pt,3.09pt) -- (1.39pt,3.16pt) -- (1.85pt,1.46pt) -- (2.17pt,-0.17pt);
  \draw[line width=0.42pt] (-0.06pt,4.41pt) circle (0.94pt);}}
\tikzset{poseD/.pic={%
  \draw[line width=0.42pt, line cap=round, line join=round] (0.00pt,0.00pt) -- (-0.40pt,-0.64pt) -- (-0.80pt,-3.31pt) -- (-0.74pt,-6.24pt) -- (-1.20pt,-6.53pt);
  \draw[line width=0.42pt, line cap=round, line join=round] (0.00pt,0.00pt) -- (0.41pt,-0.58pt) -- (0.85pt,-3.28pt) -- (0.80pt,-6.28pt) -- (0.98pt,-6.61pt);
  \draw[line width=0.42pt, line cap=round, line join=round] (0.00pt,0.00pt) -- (0.04pt,0.88pt) -- (0.06pt,1.86pt) -- (0.03pt,2.26pt) -- (-0.06pt,3.77pt) -- (0.01pt,4.39pt);
  \draw[line width=0.42pt, line cap=round, line join=round] (0.03pt,2.26pt) -- (-0.51pt,3.07pt) -- (-1.36pt,3.14pt) -- (-1.99pt,1.51pt) -- (-2.40pt,0.12pt);
  \draw[line width=0.42pt, line cap=round, line join=round] (0.03pt,2.26pt) -- (0.51pt,3.08pt) -- (1.42pt,3.15pt) -- (1.82pt,1.45pt) -- (2.18pt,-0.21pt);
  \draw[line width=0.42pt] (0.01pt,4.39pt) circle (0.94pt);}}

\newcommand{\magbars}{%
  \fill[mosKey!22] (0.000,0) rectangle (0.061,1.300);
  \fill[mosKey!22] (0.061,0) rectangle (0.122,1.282);
  \fill[mosKey!22] (0.122,0) rectangle (0.182,1.238);
  \fill[mosKey!22] (0.182,0) rectangle (0.243,1.178);
  \fill[mosKey!22] (0.243,0) rectangle (0.304,1.095);
  \fill[mosKey!22] (0.304,0) rectangle (0.365,1.011);
  \fill[mosKey!22] (0.365,0) rectangle (0.425,0.913);
  \fill[mosKey!22] (0.425,0) rectangle (0.486,0.805);
  \fill[mosKey!22] (0.486,0) rectangle (0.547,0.704);
  \fill[mosKey!22] (0.547,0) rectangle (0.608,0.605);
  \fill[mosKey!22] (0.608,0) rectangle (0.668,0.508);
  \fill[mosKey!22] (0.668,0) rectangle (0.729,0.418);
  \fill[mosKey!22] (0.729,0) rectangle (0.790,0.345);
  \fill[mosKey!22] (0.790,0) rectangle (0.851,0.274);
  \fill[mosKey!22] (0.851,0) rectangle (0.912,0.220);
  \fill[mosKey!22] (0.912,0) rectangle (0.972,0.166);
  \fill[mosKey!22] (0.972,0) rectangle (1.033,0.130);
  \fill[mosKey!22] (1.033,0) rectangle (1.094,0.093);
  \fill[mosKey!22] (1.094,0) rectangle (1.155,0.071);
  \fill[mosKey!22] (1.155,0) rectangle (1.215,0.050);
  \fill[mosKey!22] (1.215,0) rectangle (1.276,0.037);
  \fill[mosKey!22] (1.276,0) rectangle (1.337,0.026);
  \fill[mosKey!22] (1.337,0) rectangle (1.398,0.017);
  \fill[mosKey!22] (1.398,0) rectangle (1.458,0.011);
  \fill[mosKey!22] (1.458,0) rectangle (1.519,0.007);
  \fill[mosKey!22] (1.519,0) rectangle (1.580,0.005);
}
\newcommand{\sgnbars}{%
  \fill[mosVer!22] (0.000,0) rectangle (0.056,0.005);
  \fill[mosVer!22] (0.056,0) rectangle (0.111,0.008);
  \fill[mosVer!22] (0.111,0) rectangle (0.167,0.013);
  \fill[mosVer!22] (0.167,0) rectangle (0.223,0.023);
  \fill[mosVer!22] (0.223,0) rectangle (0.278,0.034);
  \fill[mosVer!22] (0.278,0) rectangle (0.334,0.050);
  \fill[mosVer!22] (0.334,0) rectangle (0.390,0.076);
  \fill[mosVer!22] (0.390,0) rectangle (0.445,0.104);
  \fill[mosVer!22] (0.445,0) rectangle (0.501,0.152);
  \fill[mosVer!22] (0.501,0) rectangle (0.557,0.202);
  \fill[mosVer!22] (0.557,0) rectangle (0.612,0.270);
  \fill[mosVer!22] (0.612,0) rectangle (0.668,0.353);
  \fill[mosVer!22] (0.668,0) rectangle (0.724,0.441);
  \fill[mosVer!22] (0.724,0) rectangle (0.779,0.549);
  \fill[mosVer!22] (0.779,0) rectangle (0.835,0.662);
  \fill[mosVer!22] (0.835,0) rectangle (0.891,0.788);
  \fill[mosVer!22] (0.891,0) rectangle (0.946,0.918);
  \fill[mosVer!22] (0.946,0) rectangle (1.002,1.028);
  \fill[mosVer!22] (1.002,0) rectangle (1.058,1.120);
  \fill[mosVer!22] (1.058,0) rectangle (1.113,1.210);
  \fill[mosVer!22] (1.113,0) rectangle (1.169,1.271);
  \fill[mosVer!22] (1.169,0) rectangle (1.225,1.300);
  \fill[mosVer!22] (1.225,0) rectangle (1.280,1.298);
  \fill[mosVer!22] (1.280,0) rectangle (1.336,1.273);
  \fill[mosVer!22] (1.336,0) rectangle (1.391,1.211);
  \fill[mosVer!22] (1.391,0) rectangle (1.447,1.131);
  \fill[mosVer!22] (1.447,0) rectangle (1.503,1.021);
  \fill[mosVer!22] (1.503,0) rectangle (1.558,0.913);
  \fill[mosVer!22] (1.558,0) rectangle (1.614,0.787);
  \fill[mosVer!22] (1.614,0) rectangle (1.670,0.671);
  \fill[mosVer!22] (1.670,0) rectangle (1.725,0.556);
  \fill[mosVer!22] (1.725,0) rectangle (1.781,0.440);
  \fill[mosVer!22] (1.781,0) rectangle (1.837,0.349);
  \fill[mosVer!22] (1.837,0) rectangle (1.892,0.267);
  \fill[mosVer!22] (1.892,0) rectangle (1.948,0.204);
  \fill[mosVer!22] (1.948,0) rectangle (2.004,0.150);
  \fill[mosVer!22] (2.004,0) rectangle (2.059,0.106);
  \fill[mosVer!22] (2.059,0) rectangle (2.115,0.075);
  \fill[mosVer!22] (2.115,0) rectangle (2.171,0.050);
  \fill[mosVer!22] (2.171,0) rectangle (2.226,0.035);
  \fill[mosVer!22] (2.226,0) rectangle (2.282,0.022);
  \fill[mosVer!22] (2.282,0) rectangle (2.338,0.014);
  \fill[mosVer!22] (2.338,0) rectangle (2.393,0.009);
  \fill[mosVer!22] (2.393,0) rectangle (2.449,0.005);
}
\newcommand{\magcurve}{\draw[mosKey, line width=0.8pt] plot[domain=0:3.4, samples=50, smooth] ({\x/3.4*1.580}, {1.300*exp(-\x*\x/2)});}
\newcommand{\sgncurve}{\draw[mosVer, line width=0.8pt] plot[domain=-3.4:3.4, samples=90, smooth] ({(\x+3.4)/6.8*2.449}, {1.300*exp(-\x*\x/2)});}

\newcommand{\ic}[2]{\tikz[baseline=-0.55ex]{\pic[draw=#1, fill=#1]{#2}}\,}
\newcommand{\icw}[1]{\tikz[baseline=-0.55ex]{\pic[draw=white, fill=white, line width=0.5pt]{#1}}\,}

\newtheorem{definition}{Definition}
\newtheorem{proposition}{Proposition}

\newcommand{\R}{\mathbb{R}}
\newcommand{\Norm}{\mathcal{N}}
\newcommand{\zm}{\mathbf{z}_{m}}
\newcommand{\zp}{\mathbf{z}_{p}}
\newcommand{\xw}{x_{\mathrm{wm}}}
\newcommand{\enc}{\mathcal{E}}
\newcommand{\dec}{\mathcal{G}}
\newcommand{\extr}{\mathcal{X}_{\phi}}
\newcommand{\lift}{\Lambda}
\newcommand{\proj}{\pi}
\newcommand{\rec}{\mathcal{R}}
\newcommand{\Hmac}{\textsf{HMAC}}
\newcommand{\Ecc}{\textsf{ECC}}
\newcommand{\Prf}{\textsf{PRF}}
\newcommand{\mact}{e\,\|\,r\,\|\,\mathrm{ctx}}
\newcommand{\efa}{\varepsilon_{\mathrm{FA}}}
\newcommand{\key}{k}
\newcommand{\dq}{d}                      
\newcommand{\Nj}{N_J}                    
\newcommand{\nb}{n}                      
\newcommand{\kb}{\kappa}                 

\begin{document}
\date{}

\title{\Large \bf MoSign: Challenge-Response Motion-Watermark Authentication\\
for Anonymous Virtual-Reality Users}

\ifarxiv
\author{
{\rm Xujun Che, Thomas Carr, Depeng Xu, Aidong Lu}\\
University of North Carolina at Charlotte\\
\{xche,tcarr23,dxu7,alu1\}@charlotte.edu
\and
{\rm Shuhan Yuan}\\
Utah State University\\
shuhan.yuan@usu.edu
} 
\fi

\maketitle

\begin{abstract}
Social virtual reality (VR) creates a paradox. A user's body motion is a high-entropy biometric: head and hand trajectories alone re-identify users among tens of thousands with over $94\%$ accuracy, so \emph{anonymizing} the rendered avatar is a practical necessity. Yet a user often still wants to \emph{prove} their identity to a chosen party from inside that anonymity. We present \textbf{MoSign}, which recasts digital watermarking as a \emph{challenge-response authentication protocol on the motion channel}. MoSign embeds a time-varying keyed message into the style latent of a motion variational autoencoder via \emph{keystream-whitened Gaussian-Shading}: watermarked motion is \emph{provably indistinguishable} from watermark-free motion, since any detector's advantage reduces to breaking a pseudorandom function, so the mark composes with anonymization. The message is a keyed MAC over an epoch counter, a session nonce, and a deployment context, making MoSign \emph{replay-resistant} and bounding forgery by the verifier's measured false-accept rate times the adversary's online query budget. A key-holding verifier decides with a sequential test. We identify render$\rightarrow$record$\rightarrow$re-estimate (``recapture'') as the realistic VR attack surface: a \emph{generic} pose estimator strips the necessarily subtle watermark, but a \emph{recapture-robust keyed reader} recovers it (up to $0.96$ codeword accuracy on a projected-2D channel, $0.81$ through a full render-to-video loop), while without the key recovery stays at chance. On HumanML3D, MoSign authenticates every legitimate user at a false-accept rate of $10^{-4}$ on clean and most channels and stays undetectable (detection AUC $0.51$, chance $0.5$); on the BOXRR-23 VR dataset it carries the mark through a real anonymizer at $0.99$ codeword accuracy and adds no de-anonymization side channel.
\end{abstract}

\section{Introduction}
Immersive social VR places \emph{anonymous} avatars in shared spaces, where two
requirements collide.
\emph{First, motion is identity.} Large-scale measurement shows that
head-and-hand telemetry uniquely identifies more than $50{,}000$ users, reaching
$94.33\%$ accuracy from $100$ seconds and $73.20\%$ from only $10$ seconds of
motion~\cite{nair2023unique}; the same signal leaks over forty private attributes
spanning demographics, behaviour, and health~\cite{nair2023inferring} and is available at population scale in
public datasets~\cite{nair2024boxrr}; identification is reliable within one application,
though it generalizes across applications only weakly~\cite{schach2025crossxr}. Consequently, real-time motion
\emph{anonymizers} are now being built and deployed~\cite{nair2024deepmasking}.
\emph{Second, anonymity defeats accountability.} Once a user is anonymized, they
can no longer prove ``I am the same trusted agent you spoke to before'' to a
chosen counterparty, a private service, or a moderation system, without revealing
\emph{who} they are to every bystander. This is the missing primitive:
authentication \emph{within} anonymity.

Existing watermarking cannot fill this role. Generative watermarks for images,
video, and 3D assets~\cite{fernandez2023stablesignature,wen2023treering,
yang2024gaussianshading,jang2024lvmark,li2025gaussianseal} are built for
\emph{provenance}: they embed a \emph{static} payload and assume an adversary who
edits the carrier, not one who \emph{records and replays} it. A static watermark
on a moving body is, in effect, a password the body continuously broadcasts: any
bystander can capture it and impersonate its owner. Moreover, no systematic
in-generation watermark exists for the skeletal-motion modality at all. Classical
motion-capture watermarking~\cite{agarwal2007tamper,li2007progressive,
motwani2008skinning,du2012maxima} is post-hoc, hand-crafted, decoupled from any
generator, and aimed at copyright or at detecting that a clip has been altered, not at
proving that a live party holds a secret.

Nor do the established ways of proving something without revealing who you are. Anonymous
credentials~\cite{chaum1985security,camenisch2001credentials} authenticate the holder of a credential
when a session is established, but that binds the handshake rather than what subsequently travels over
the channel, so it cannot say whether the motion arriving now is the motion the credential holder
produced. VR behavioural biometrics~\cite{pfeuffer2019behavioural,miller2020vrauth} authenticate from
motion directly, by recognizing the person, which is exactly the capability anonymization removes.
Neither carries the proof \emph{in the stream itself}; Appendix~\ref{sec:related} treats these
lines of work in full.

MoSign closes this gap by treating the watermark as a \emph{protocol}, not a
label. The prover, the anonymous user proving who they are, embeds a time-varying
keyed message into the style latent of a generated motion in a way that keeps the
motion statistically identical to ordinary watermark-free motion; the verifier,
the single party that user chose to share a key with, reads the message and
authenticates with a sequential hypothesis test. Crucially, we take the
\emph{realistic} VR observation model seriously: a third party sees only the
\emph{rendered} avatar and must re-estimate motion from video, the ``analog
hole'' of VR. We show this channel, not benign signal processing, is the true
adversary, and we give a clean, security-grounded answer for it.

\begin{figure*}[tp]\centering
\begin{tikzpicture}[mosfig, font=\scriptsize]
  \begin{scope}[on background layer]
    \fill[mosNeu!8, rounded corners=3pt] (4.05,-1.15) rectangle (11.85,1.15);
    \fill[mosAcc!7, rounded corners=3pt] (7.95,-1.15) rectangle (11.85,1.15);
  \end{scope}
  \node[lab, anchor=north west] at (4.18,1.08) {public motion channel, controlled by the adversary};

  \cardn{P}{mosKey}{29mm}{(1.05,0)}{15mm}{\icw{keyicon}Prover, holds $\key$}{writes $\Hmac_{\key}(\mact)$\\into the style latent}
  \cardn{A}{mosNeu}{25mm}{(5.75,0)}{15mm}{\icw{screen}Public VR}{anonymized avatar,\\identity masked}
  \cardn{R}{mosAcc}{25mm}{(9.95,0)}{15mm}{\icw{cam}Recapture}{render, record,\\re-estimate pose}
  \cardn{V}{mosVer}{29mm}{(14.35,0)}{15mm}{\icw{keyicon}Verifier, holds $\key$}{keyed reader, SPRT,\\\textbf{accept}}

  \bits{mosKey}{0.19}{-1.18}{0.0955}{0.24}{18}{1,0,1,1,0,1,0,0,1,0,1,1,0,1,0,1,1,0}
  \node[klab, anchor=north] at (1.05,-1.30) {the codeword written under $\key$};
  \bits{mosVer}{13.49}{-1.18}{0.0955}{0.24}{18}{1,0,1,1,0,1,0,0,1,0,1,1,0,1,0,1,1,0}
  \node[lab, anchor=north, text=mosVer] at (14.35,-1.30) {the same codeword, recovered};

  \draw[flow, line width=0.9pt] (P) -- (A);
  \foreach \d/\k in {0/poseA, 0.32/poseB, 0.64/poseC}{\pic[draw=mosKey!85] at (3.02+\d,0.42) {\k};}
  \node[elab, anchor=north, font=\tiny] at (3.34,-0.10) {one $\tau$ epoch};
  \draw[flow, line width=0.9pt] (A) -- (R);
  \draw[flow, line width=0.9pt] (R) -- node[elab, above]{recovered} node[elab, below]{motion} (V);

  \draw[keyflow, line width=0.9pt] (V.north) -- ++(0,7mm) -| (P.north);
  \node[klab, fill=white, inner sep=1.8pt, rounded corners=1pt] at (7.7,1.78)
       {challenge $(e,r)$: a fresh nonce and the epoch counter, renewed every $\tau$};

  \draw[advflow, line width=0.8pt] (7.9,-1.15) -- (7.9,-1.62);
  \begin{scope}[on background layer]
    \fill[mosAdv!5, rounded corners=3pt] (-0.45,-4.28) rectangle (15.80,-1.62);
  \end{scope}
  \node[draw=mosAdv, fill=mosAdv, text=white, font=\tiny\bfseries, rounded corners=2pt,
        inner sep=2.8pt, anchor=north west] at (-0.30,-1.72)
       {\icw{eye}The same band, to a party that does not hold $\key$};
  \draw[draw=mosAdv!25, line width=0.4pt] (5.55,-2.52) -- (5.55,-4.12);
  \draw[draw=mosAdv!25, line width=0.4pt] (10.05,-2.52) -- (10.05,-4.12);
  \foreach \x/\h in {0.20/{it holds}, 5.85/{its best detector reaches},
                     10.30/{so it cannot}}{
    \node[font=\tiny\bfseries, text=mosAdv, anchor=west] at (\x,-2.30) {\h};}

  \foreach \x/\lab in {1.45/{signed}, 4.05/{unsigned}}{
    \node[draw=mosNeu!35, fill=white, rounded corners=1.8pt,
          minimum width=21mm, minimum height=5.6mm] at (\x,-2.98) {};
    \foreach \d/\k in {-0.585/poseA, -0.195/poseB, 0.195/poseC, 0.585/poseD}{
      \pic[draw=mosNeu!88] at (\x+\d,-2.98) {\k};}
    \node[font=\tiny, text=mosAdv] at (\x,-3.44) {\lab};}
  \node[font=\scriptsize, text=mosNeu] at (2.75,-2.98) {$\equiv$};
  \node[font=\tiny, text=mosAdv, align=left, text width=50mm, anchor=north west] at (0.20,-3.60)
       {signed and unsigned motion are identically distributed, so no test separates them};

  \begin{scope}[shift={(5.90,-3.06)}]
    \def\sx{6.545}                       
    \fill[mosAdv!7, rounded corners=1.2pt] ({(0.49-0.45)*\sx},-0.13) rectangle ({(0.54-0.45)*\sx},0.13);
    \draw[draw=mosNeu!55, line width=0.5pt] (0,0) -- ({(1.0-0.45)*\sx},0);
    \foreach \v in {0.5,0.75,1.0}{
      \draw[draw=mosNeu!55, line width=0.5pt] ({(\v-0.45)*\sx},0) -- ({(\v-0.45)*\sx},-0.09);
      \node[font=\tiny, text=mosNeu, anchor=north] at ({(\v-0.45)*\sx},-0.11) {$\v$};}
    \draw[draw=mosNeu!70, densely dashed, line width=0.5pt]
      ({(0.5-0.45)*\sx},0.02) -- ({(0.5-0.45)*\sx},0.44);
    \draw[draw=mosAdv, line width=0.9pt, line cap=round]
      ({(0.49-0.45)*\sx},0) -- ({(0.54-0.45)*\sx},0);
    \foreach \v in {0.49,0.54}{
      \draw[draw=mosAdv, line width=0.9pt] ({(\v-0.45)*\sx},-0.09) -- ({(\v-0.45)*\sx},0.09);}
    \fill[mosAdv] ({(0.51-0.45)*\sx},0) circle (1.7pt);
    \node[font=\tiny\bfseries, text=mosAdv, anchor=west] at ({(0.545-0.45)*\sx},0.30)
         {AUC $0.51$ \textnormal{(95\%)}};
    \node[font=\tiny, text=mosNeu!85, anchor=north] at ({(0.5-0.45)*\sx},-0.40) {chance};
    \node[font=\tiny, text=mosNeu!85, anchor=north east] at ({(1.0-0.45)*\sx},-0.40) {perfect};
  \end{scope}

  \begin{scope}[shift={(10.30,-2.84)}]
    \foreach \i/\t/\c in {0/{tell it is signed}/{same law as the prior},
                          1/{link two sessions}/{renewed each epoch},
                          2/{forge a response}/{needs $\key$}}{
      \node[draw=mosAdv!45, fill=white, rounded corners=2pt, inner sep=0pt,
            minimum width=52mm, minimum height=5.2mm, anchor=west] at (0,-\i*0.56) {};
      \draw[draw=mosAdv!85, line width=0.9pt, line cap=round]
        (0.17,-\i*0.56-0.115) -- (0.40,-\i*0.56+0.115)
        (0.17,-\i*0.56+0.115) -- (0.40,-\i*0.56-0.115);
      \node[font=\tiny, text=mosAdv, anchor=west] at (0.55,-\i*0.56) {\t};
      \node[font=\tiny, text=mosAdv!70, anchor=east] at (5.10,-\i*0.56) {\c};}
  \end{scope}
\end{tikzpicture}
\caption{MoSign as a challenge-response protocol carried on the motion channel. The verifier issues a fresh
challenge, an epoch counter $e$ and a session nonce $r$; the prover answers not with a message but
with one epoch of motion that carries
$\Hmac_{\key}(\mact)$ in its style latent. The response survives the public channel, including a
render-record-re-estimate recapture, and is read only by a party holding $\key$; to everyone else it is
indistinguishable from unmarked anonymized motion. Key glyphs mark the two parties that hold $\key$,
and the two pose sequences are the same motion drawn twice.}
\label{fig:overview}
\end{figure*}
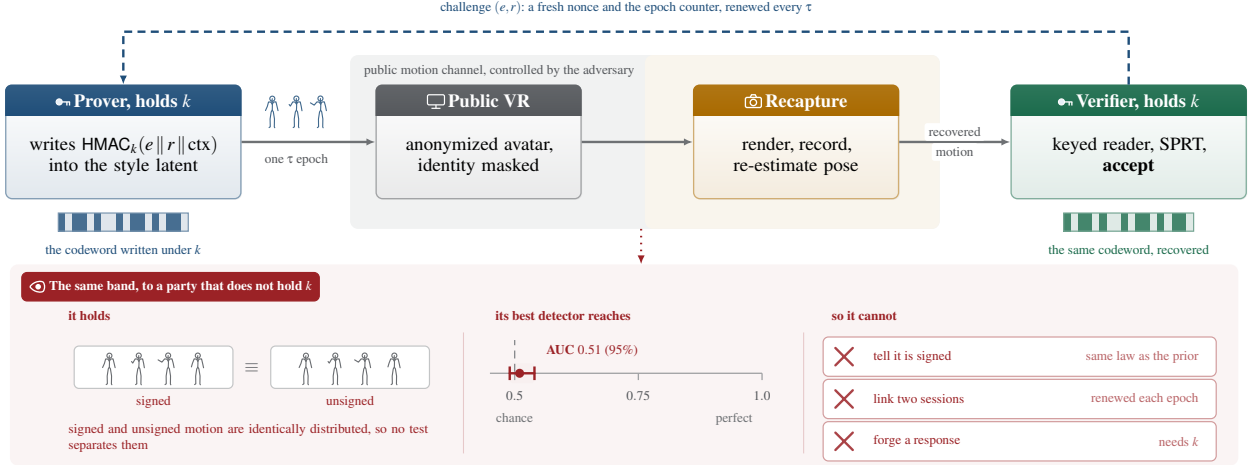
\smallskip\noindent\textbf{Contributions.}
\begin{itemize}\itemsep3pt
\item[\textbf{C1}] \emph{The first in-generation watermark for skeletal motion,
and, to our knowledge, the first used for \emph{entity} authentication, proving that
a live party holds a secret, rather than for provenance or content integrity.}
Unlike post-hoc mocap watermarks~\cite{agarwal2007tamper,du2012maxima}, the mark is produced
\emph{during} generation and read back through the full VR pipeline, and it answers ``is this the
key holder?'' rather than ``where did this come from?''.
\item[\textbf{C2}] \emph{A challenge-response, time-varying watermark that makes
the scheme replay-resistant and forgery-resistant to a bounded-query adversary.} We embed
$\Hmac_{\key}(\text{epoch}\,\|\,\text{nonce}\,\|\,\text{ctx})$ rather than a static
key, turning a ``broadcast password'' into an authentication protocol on the
motion channel. To our knowledge this is the first watermark with these protocol
guarantees, and it directly answers the recording-and-replay adversary that
defeats static schemes.
\item[\textbf{C3}] \emph{Provable undetectability, hence composability with
anonymization.} Keystream-whitened
Gaussian-Shading makes the watermarked style latent identically distributed to a
fresh prior sample, so the area under the detector's ROC curve (AUC) is $0.5$ \emph{by construction} (for any
decoder and any discriminator), reducing to PRF security. The mark thus never
betrays that a user is ``signing'' and stacks on top of any anonymizer.
\item[\textbf{C4}] \emph{A recapture-centered threat model and a receiver-side
keyed reader that solves it.} We show a generic estimator
strips the (subtle, undetectable) watermark, expose this as an inherent
undetectability/robustness tension, and resolve it by exploiting that the verifier \emph{holds the key}: a recapture-robust keyed reader recovers the
watermark from recaptured video, with security reducing to the keyed whitening
(without the key, recovery is provably at chance).
\end{itemize}

\noindent Figure~\ref{fig:overview} overviews the protocol. The design and its security analysis follow;
the evaluation then substantiates each claim with measurements.


\section{Background and Threat Model}\label{sec:threat}
\noindent\textbf{Motion biometrics and anonymization.}
VR motion is a strong biometric at population scale~\cite{nair2023unique,
nair2024boxrr}, reliable within an application though generalizing across applications only
weakly~\cite{schach2025crossxr}, and is recoverable even from a \emph{screen recording} of an
avatar's gait, robustly to appearance changes~\cite{meng2024avatarhunter}. Anonymizers therefore perturb the
identity-bearing factors of motion~\cite{nair2024deepmasking}. MoSign
operates after, and orthogonally to, anonymization: it adds an authentication
signal that neither de-anonymizes the user nor is detectable as ``signing.'' The
party a user authenticates to and the parties it stays anonymous to are
\emph{distinct}: a keyed mark is by design linkable by the key holder, so
authentication is to a \emph{chosen, key-sharing} verifier, not the anonymous crowd
(Figure~\ref{fig:deploy}).

\smallskip\noindent\textbf{Adversary (Kerckhoffs).}
The adversary knows the entire system, including architectures, weights, and the
training procedure, but not the secret key $\key$. We require undetectability and
unforgeability to hold against this adversary. Our key-necessity experiments instantiate exactly this adversary: it is even allowed to
train its own readers from scratch. It may also \emph{interact} with the verifier,
submitting candidate responses and learning which are accepted; as for any
challenge-response protocol, we assume the verifier rate-limits this, and we measure how success grows with the query budget. Finally it
may be \emph{in line} on either the motion or the video channel, modifying what the
verifier receives rather than only observing it.

Two experiments go beyond this model to isolate where security comes from, giving
the adversary key material it would not otherwise have: the dewhitening keystream
without the reader's key conditioning, and the per-session key for a white-box removal attack.

\smallskip\noindent\textbf{Channels.}
Beyond benign processing (re-sampling, smoothing, sensor noise, retargeting,
cropping, joint dropout), the defining VR channel is \emph{recapture}: a deployable
verifier often does not receive the raw joint stream and observes only the
\emph{rendered} avatar, which it screen-records and runs through a monocular video
pose estimator to recover the motion~\cite{meng2024avatarhunter}. Authentication
must therefore survive this render-then-re-estimate loop for the key holder, while
a keyless party can neither detect nor forge the mark on either the raw or the
recaptured stream.

\smallskip\noindent\textbf{Security goals.}
(i) \emph{Undetectability}: a keyless party cannot tell watermarked from ordinary
anonymized motion. (ii) \emph{Unforgeability}: a keyless party cannot make a
verifier accept. (iii) \emph{Replay resistance}: a recorded session does not
authenticate outside the $\pm$slack freshness window, which is tunable toward zero.
(iv) \emph{Robustness}: authentication survives the channels above, including
recapture, for the key holder. (v) \emph{Integrity under modification}: an in-line
adversary can suppress authentication by distorting the motion, but cannot turn suppression into acceptance.

\smallskip\noindent\textbf{Nonce and epoch.}
The verifier issues a fresh nonce per session as its challenge, and the epoch is self-clocked from a shared
wall clock with a $\pm$slack-epoch tolerance; the codeword and whitening keystream both derive from
$\Hmac_{\key}(\text{epoch}\,\|\,\text{nonce}\,\|\,\text{ctx})$, so a response is valid only for its
$(\text{epoch},\text{nonce})$. Freshness of the verifier-issued nonce is what defeats cross-session replay; the $\pm$slack window is the only surface a relay can exploit, and it is bounded.
Nonce hygiene (one challenge nonce per session, never reused) is load-bearing: a reused nonce revives replay,
which we quantify.

\smallskip\noindent\textbf{Deployment: three cases.}
The verifier is a party the user \emph{chooses} to authenticate to and shares a key with, distinct from the
parties it stays anonymous to. Any observer of the avatar's motion falls into one of three cases
(Figure~\ref{fig:deploy}):
\begin{itemize}\itemsep2pt
\item[\textbf{1}] \emph{Keyless party} (other users, platform analytics, an adversary, or the game/platform
server itself, which receives the anonymized stream but no authentication key): sees the motion, raw or
recaptured, but holds no key. It cannot tell the user is signing, link the user across sessions through the mark,
or impersonate the user, so the user is anonymous and unforgeable to it (C3, C2).
\item[\textbf{2}] \emph{Key-holding verifier, raw stream} (a cross-application service reached through the
platform's SDK or a relay): reads the mark on a clean channel and authenticates the pseudonymous user (C1, C2).
\item[\textbf{3}] \emph{Key-holding verifier, video only} (a moderator, an event gate, or a peer who sees the
rendered avatar but not its telemetry): screen-records the avatar, re-estimates its motion, and reads the mark through recapture; the hard case MoSign uniquely handles (C4).
\end{itemize}
The same entity can occupy different cases: the platform server is the prototypical case-1 party, yet a user may
\emph{selectively} authenticate to that operator (moving it to case 2, for example to prove account ownership),
accepting linkage to that one party while staying anonymous to every other. The verifier is therefore never the
party the user hides from \emph{in the same act}: a key holder can link the keyed mark by design, so seeking
anonymity against the very service one authenticates to would be self-defeating.

Verification is symmetric, and deliberately so: the verifier holds $\key$ and can produce the same
transcript, so a response convinces that verifier and no third party. A transferable proof would let
the one verifier a user authenticates to prove that user's identity to everyone else, the very
linkage the design prevents, so the authentication is deniable by construction. MoSign establishes
membership in a shared-key relationship, not non-repudiable attribution.

\begin{figure*}[tp]\centering
\begin{tikzpicture}[mosfig, font=\scriptsize]
  \foreach \d/\k in {-0.42/poseA, -0.14/poseB, 0.14/poseC, 0.42/poseD}{
    \pic[draw=mosNeu!88] at (\d,0.10) {\k};}
  \node[font=\scriptsize, anchor=north] at (0,-0.16) {real motion};

  \node[ghost, text width=22mm, minimum height=9.5mm] (an) at (2.42,0)
       {\ic{mosNeu!55}{anon}anonymizer\\(upstream)};
  \cardn{ms}{mosKey}{28mm}{(5.81,0)}{11.5mm}{MoSign}{keyed watermark $(e,r)$}
  \node[plain, text width=26mm, minimum height=12mm] (str) at (9.40,0) {};
  \node[font=\scriptsize, anchor=north] at (9.40,0.52) {\ic{mosNeu}{screen}render avatar};
  \foreach \d/\k in {-0.48/poseA, -0.16/poseB, 0.16/poseC, 0.48/poseD}{
    \pic[draw=mosNeu!88] at (9.40+\d,-0.12) {\k};}

  \draw[flow, line width=0.85pt] (0.66,0) -- (an.west);
  \draw[flow, line width=0.85pt] (an)--(ms);
  \draw[flow, line width=0.85pt] (ms)--(str);

  \cardn{obs}{mosAdv}{46mm}{(14.24,1.32)}{11mm}{\icw{eye}\ding{202}\ any party without $\key$}{\hspace*{6mm}cannot detect, link, or forge}
  \cardn{vraw}{mosVer}{46mm}{(14.24,0)}{11mm}{\icw{keyicon}\ding{203}\ verifier holding $\key$, raw stream}{\hspace*{6mm}reads the codeword directly}
  \cardn{vvid}{mosAcc}{46mm}{(14.24,-1.32)}{11mm}{\icw{keyicon}\ding{204}\ verifier holding $\key$, video only}{\hspace*{6mm}recaptures, then decodes}

  \coordinate (j) at (11.22,0);
  \draw[draw=mosNeu!85, line width=0.85pt] (str.east) -- (j);
  \draw[draw=mosNeu!85, line width=0.85pt] (11.22,1.32) -- (11.22,-1.32);
  \foreach \n/\y in {obs/1.32, vraw/0, vvid/-1.32}{
    \draw[flow, line width=0.85pt] (11.22,\y) -- (\n.west);}
  \fill[mosNeu] (j) circle (1.1pt);

  \foreach \d/\k/\o in {-0.13/poseB/60, 0.13/poseD/28}{\pic[draw=mosAdv!\o] at (12.30+\d,1.16) {\k};}
  \node[font=\tiny\bfseries, text=mosAdv] at (12.55,1.34) {?};
  \bits{mosVer}{11.99}{-0.25}{0.056}{0.20}{11}{1,0,1,1,0,1,0,0,1,0,1}
  \draw[mosAcc!65, line width=0.45pt, rounded corners=0.6pt] (12.00,-1.62) rectangle (12.60,-1.28);
  \foreach \d/\k/\o in {-0.13/poseC/70, 0.13/poseA/40}{\pic[draw=mosAcc!\o] at (12.30+\d,-1.46) {\k};}
\end{tikzpicture}
\caption{Deployment, and the three parties that see the same rendered avatar: \ding{202} keyless,
\ding{203} key holder with the raw stream, \ding{204} key holder with only the rendered video. A
separate upstream anonymizer masks identity; MoSign adds only the keyed watermark, so the two compose
without either being redesigned for the other.}
\label{fig:deploy}
\end{figure*}
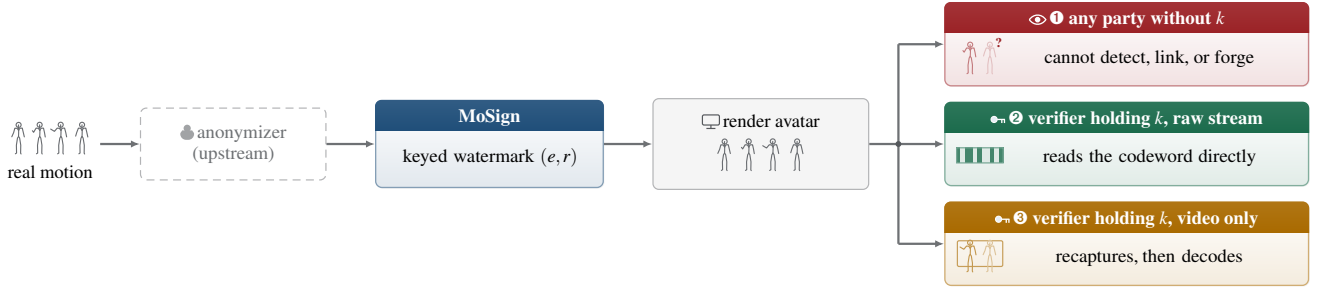
\section{Design}\label{sec:design}
\subsection{Notation and the motion VAE}
Motion is a sequence $x\in\R^{T\times D}$ of $D$ per-frame features in a redundant,
root-relative parameterization, from which a differentiable map $\rec(\cdot)$ recovers the
$\Nj$ joint positions $J\in\R^{T\times\Nj\times 3}$. Nothing in this section depends on the
particular feature set, skeleton, or capture rate; the implementation section gives the ones we
instantiate. Windows are stated in seconds throughout and converted to frames at whatever rate
the capture runs at: the verifier's sliding window is $2$\,s and an epoch is $\tau{=}10$\,s. A transformer VAE
$(\enc,\dec)$ factorizes motion into a \emph{content} latent
$\zm\in\R^{(T/4)\times\dq}$ and a per-window \emph{style} latent $\zp\in\R^{\dq}$
($\dq{=}256$) with prior $\zp\sim\Norm(0,I_\dq)$; the decoder $\dec$ conditions on
$\zp$ through per-layer adaptive normalization. Thus $\zp$ modulates global style
and serves as the \emph{watermark carrier}, while $\zm$ carries the temporal
content.

\emph{An identity-neutral carrier.} The watermark overwrites only the style latent
$\zp$ with a keyed, user-independent draw, so the deployed $\zp$ introduces neither the
user's identity nor a privacy leak. Any residual identity remains in the content latent
$\zm$, which the watermark never alters; removing it is the job of an upstream
anonymizer, a separate problem, and one on which the composition adds no
de-anonymization side channel.

\subsection{Keyed embedding (prover)}\label{sec:embed}
For epoch $e$, session nonce $r$, key $\key$, and a deployment context string
$\mathrm{ctx}$ that binds the mark to one application and verifier, the prover forms a
time-varying mark in four steps (Fig.~\ref{fig:embed}):
\begin{align}
m_e &= \Hmac_{\key}(\mact)[:\kb] \in\{0,1\}^{\kb},\\
c   &= \Ecc(m_e)\in\{0,1\}^{\nb} \quad(\mathrm{BCH}(\nb,\kb,t)),\\
b   &= \mathrm{rep}(c)\oplus \Prf_{\key}(\mact)\in\{0,1\}^{\dq},\\
\zp[i] &= |\zeta_i|\,(2\,b_i-1),\quad \zeta\sim\Norm(0,I_\dq)\ \text{(keyed)},
\end{align}
Here $\Ecc$ is a Bose--Chaudhuri--Hocquenghem (BCH) error-correcting code, $\Prf$ a pseudorandom
function, and $\mathrm{rep}$ a repetition code; we instantiate $\nb{=}63$, $\kb{=}30$, and $t{=}6$ correctable errors. We write
$w=\Prf_{\key}(\mact)$ for the keystream and $s_i=2b_i-1\in\{\pm1\}$ for the embedded sign.
$\Hmac$ and $\Prf$ are not two independent objects sharing a key. Both are invocations of one
PRF under $\key$ on disjoint label prefixes, $\texttt{msg}\|\cdot$ for the message and
$\texttt{mask}\|\cdot$ for the keystream, so their outputs are the images of disjoint parts of a
single domain and the analysis below reduces to that one PRF rather than assuming independence
between two.

\emph{(1) Time-varying message.} The payload is not the key but a keyed MAC over
the current epoch and the session nonce. This single choice is what upgrades the
scheme from a static watermark to an authentication protocol: a recording made at
epoch $e$ carries $m_e$, which is invalid at any later epoch, so replay fails by construction.

\emph{(2) Error correction.} The motion channel has low bandwidth (about $22$
joints over a short window), so we protect the message with a $\mathrm{BCH}(63,30,t{=}6)$
code. The verifier compares in the codeword domain and need not decode, which both
amplifies robustness and lets us report an exact analytic null.

\emph{(3) Keystream whitening.} $\mathrm{rep}(\cdot)$ tiles the codeword to width
$\dq$ (about four replicas per bit, a natural repetition code) and XORs a keyed
PRF stream. To a keyless observer the embedded sign pattern $b$ is therefore
pseudorandom and i.i.d.; this is the property that makes the next step
distribution-preserving and, crucially, makes \emph{recovery} require the key.

\emph{(4) Gaussian-Shading.} Following~\cite{yang2024gaussianshading}, because the
magnitude $|\zeta_i|$ is independent of the sign, setting
$\zp[i]=|\zeta_i|(2b_i{-}1)$ leaves $\zp$ \emph{marginally exactly}
$\Norm(0,I_\dq)$. The message lives only in the signs; the magnitudes carry none.
The watermarked motion is $\xw=\dec(\zm,\zp)$.

\paragraph{Symmetry requirement.}
The construction is deliberately \emph{symmetric}: every operation applied to
watermarked motion is also applied to the watermark-free baseline, and there is
\emph{no} generation-only step. This is precisely what makes undetectability exact. Adding even one asymmetric, generation-only operation
re-introduces a detectable footprint, as we confirm by measurement: a bounded
amplitude-correction term $\Delta\theta(\zm,c)$ of magnitude $\delta$ added to
$\tilde\zp$, or an auxiliary objective that imprints the message into joint
positions. We ship $\delta{=}0$, so no such term is present, and pay the cost on the
receiver side with a stronger extractor and a soft-combining read-out.

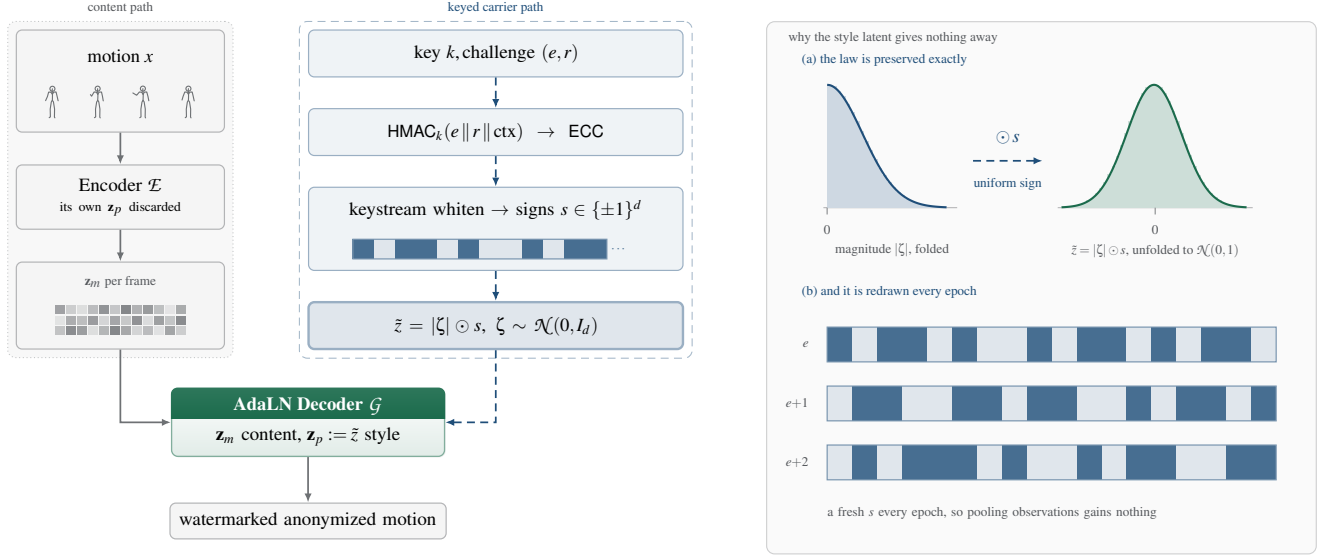
\begin{figure*}[tp]\centering
\begin{tikzpicture}[mosfig]
  \tikzset{ln/.style={text width=25mm, minimum height=8.5mm},
           rn/.style={text width=47mm, minimum height=6.2mm}}
  \def\xR{4.966}   
  \def\xM{2.483}   
  \node[plain, ln, minimum height=13.5mm] (x) at (0,2.275) {};
  \node[font=\scriptsize, anchor=north] at (0,2.82) {motion $x$};
  \foreach \d/\k in {-0.90/poseA, -0.30/poseB, 0.30/poseC, 0.90/poseD}{\pic[draw=mosNeu!85] at (\d,2.02) {\k};}
  \node[plain, ln] (enc) at (0,0.74) {Encoder $\enc$\\[-0.5pt]{\tiny its own $\zp$ discarded}};
  \node[plain, ln, minimum height=11.5mm] (zm) at (0,-0.695) {};
  \node[lab, anchor=north] at (0,-0.25) {$\zm$ per frame};
  \foreach \r in {0,1,2}{\foreach \c/\t in {0/40,1/14,2/62,3/24,4/48,5/18,6/56,7/30,8/44,9/22,10/58,11/34}{
    \pgfmathsetmacro{\tt}{mod(\t+\r*19,58)+16}
    \fill[mosNeu!\tt] (-0.86+\c*0.145,-0.81-\r*0.135) rectangle (-0.735+\c*0.145,-0.695-\r*0.135);}}

  \node[card={mosKey}{47mm}, rn] (k) at (\xR, 2.64) {key $\key$,\ challenge $(e,r)$};
  \node[card={mosKey}{47mm}, rn] (h) at (\xR, 1.60) {$\Hmac_{\key}(\mact)\ \to\ \Ecc$};
  \node[card={mosKey}{47mm}, minimum height=11mm, text width=47mm] (w) at (\xR,0.32) {};
  \node[font=\scriptsize, text width=47mm, align=center, anchor=north, inner sep=0pt]
       at ([yshift=-1.6mm]w.north) {keystream whiten $\to$ signs $s\in\{\pm1\}^{\dq}$};
  \begin{scope}[shift={(w.south)}]
    \bits{mosKey}{-1.89}{0.16}{0.28}{0.26}{12}{1,0,1,1,0,1,0,0,1,0,1,1}
    \node[font=\tiny, text=mosKey!80] at (1.63,0.29) {$\cdots$};
  \end{scope}
  \node[card={mosKey}{47mm}, rn, line width=0.9pt, fill=mosKey!14] (z) at (\xR,-0.96)
       {$\tilde z=|\zeta|\odot s$,\ \ $\zeta\sim\Norm(0,I_\dq)$};

  \cardn{dec}{mosVer}{34mm}{(\xM,-2.24)}{9mm}{AdaLN Decoder $\dec$}{$\zm$ content,\ $\zp\!:=\!\tilde z$ style}
  \node[plain, text width=34mm] (out) at (\xM,-3.545)
       {watermarked anonymized motion};

  \begin{scope}[on background layer]
    \node[chan, fit=(x)(enc)(zm), inner sep=3pt] (LB) {};
    \node[scope, fit=(k)(z), inner sep=3pt] (RB) {};
  \end{scope}
  \node[lab,  above=0.4mm of LB.north] {content path};
  \node[klab, above=0.4mm of RB.north] {keyed carrier path};

  \draw[flow] (x) -- (enc);
  \draw[flow] (enc) -- (zm);
  \draw[keyflow] (k) -- (h); \draw[keyflow] (h) -- (w); \draw[keyflow] (w) -- (z);
  \draw[flow]    (zm.south) |- (dec.west);
  \draw[keyflow] (z.south)  |- (dec.east);
  \draw[flow] (dec) -- (out);

  \begin{scope}[shift={(9.35,0.60)}]
    \draw[mosNeu!45, line width=0.4pt] (-0.05,0) -- (1.72,0);
    \begin{scope}[yscale=1.25] \magbars \magcurve \end{scope}
    \draw[mosNeu!60, line width=0.5pt] (0,0) -- (0,-0.09);
    \node[font=\tiny, text=mosNeu, anchor=north] at (0,-0.10) {$0$};
    \node[font=\tiny, text=mosNeu, anchor=north] at (0.86,-0.34) {magnitude $|\zeta|$, folded};
    \draw[keyflow] (1.93,0.62) -- (2.85,0.62);
    \node[font=\scriptsize, text=mosKey, anchor=south] at (2.39,0.72) {$\odot\,s$};
    \node[font=\tiny, text=mosKey, anchor=north] at (2.39,0.52) {uniform sign};
    \begin{scope}[shift={(3.10,0)}]
      \draw[mosNeu!45, line width=0.4pt] (-0.05,0) -- (2.53,0);
      \begin{scope}[yscale=1.25] \sgnbars \sgncurve \end{scope}
      \draw[mosNeu!60, line width=0.5pt] (1.24,0) -- (1.24,-0.09);
      \node[font=\tiny, text=mosNeu, anchor=north] at (1.24,-0.10) {$0$};
      \node[font=\tiny, text=mosNeu, anchor=north] at (1.24,-0.34)
           {$\tilde z=|\zeta|\odot s$, unfolded to $\Norm(0,1)$};
    \end{scope}
  \end{scope}

  \begin{scope}[shift={(9.35,-2.22)}]
    \node[font=\tiny, text=mosNeu, anchor=east] at (-0.12,1.01) {$e$};
    \bits{mosKey}{0}{0.78}{0.33}{0.46}{18}{1,0,1,1,0,1,0,0,1,0,1,1,0,1,0,1,1,0}
    \node[font=\tiny, text=mosNeu, anchor=east] at (-0.12,0.23) {$e{+}1$};
    \bits{mosKey}{0}{0.00}{0.33}{0.46}{18}{0,1,1,0,0,1,1,0,1,1,0,0,1,0,1,1,0,1}
    \node[font=\tiny, text=mosNeu, anchor=east] at (-0.12,-0.55) {$e{+}2$};
    \bits{mosKey}{0}{-0.78}{0.33}{0.46}{18}{0,1,0,1,1,1,0,1,0,0,1,0,1,1,0,0,1,1}
    \node[font=\tiny, text=mosNeu, anchor=north west, align=left, text width=54mm] at (-0.12,-1.02)
         {a fresh $s$ every epoch, so pooling observations gains nothing};
  \end{scope}

  \begin{scope}[on background layer]
    \node[zone, draw=mosNeu!30, fill=mosNeu!3,
          fit={(8.62,-3.91) (15.75,2.99)}, inner sep=2pt] (PN) {};
  \end{scope}
  \node[lab, anchor=north west] at (8.78,2.98) {why the style latent gives nothing away};
  \node[klab, anchor=west] at (8.96,2.56) {(a) the law is preserved exactly};
  \node[klab, anchor=west] at (8.96,-0.52) {(b) and it is redrawn every epoch};
\end{tikzpicture}
\caption{Keyed, distribution-preserving embedding into the style latent. \emph{Left:} content and keyed
carrier travel separately and meet only at the decoder, so the message never enters $\zm$ and an
upstream anonymizer cannot interfere with it. \emph{Right:} the substitution is exact, because $|\zeta|$ is a folded
normal and the keyed sign vector is uniform and independent of it
(Proposition~\ref{prop:wmind}). The histograms are measured from the deployed codec: over $10^6$
written coordinates a Kolmogorov--Smirnov test against the prior gives $D{=}0.0019$ ($p{=}0.44$).
Poses throughout the figures are HumanML3D motion.}
\label{fig:embed}
\end{figure*}
\subsection{Reading and verification (verifier)}\label{sec:read}
A learned extractor $\extr:J\mapsto\boldsymbol{\ell}\in\R^{\dq}$ produces per-sign
logits from observed joints. It canonicalizes $J$ (removes root XZ translation and
normalizes by mean bone length), forms rich frame features (position, velocity,
acceleration, and bone directions), takes the low-frequency subband of a temporal discrete wavelet transform (DWT)
(where the watermark energy concentrates, following the temporal-robustness idea
of~\cite{jang2024lvmark}), and applies a transformer. The key holder then
\emph{dewhitens and combines}:
\begin{equation}
\hat{c}_j \;=\!\!\sum_{i\,:\,i\equiv j \;(\mathrm{mod}\ \nb)}\!\!
   \ell_i\,\bigl(1-2\,\Prf_{\key}(\mact)_i\bigr),\quad j=1,\dots,\nb,
\label{eq:comb}
\end{equation}
that is, it XORs the keystream back and soft-sums the roughly four replicas of each
codeword bit. Equation~\eqref{eq:comb} \emph{requires} the key, and it is where
the redundancy becomes coherent: without the keystream the replicas are flipped by
independent pseudorandom bits and cancel.

\paragraph{Sequential decision.}
Over a sliding window the verifier counts matches between $\mathrm{sign}(\hat c)$
and the expected codeword $c^\star=\Ecc(\Hmac_{\key}(\mact))$ against
the null ``no watermark'' (match count distributed as
$\mathrm{Binom}(\nb,\tfrac12)$) and the alternative that the window carries the mark at a
per-window codeword accuracy $p_1$, which we calibrate per channel. Following the hypothesis-testing view of
watermark detection~\cite{kirchenbauer2023watermark}, per-window log-likelihood
ratios (LLRs) feed Wald's sequential probability ratio test (SPRT): accept when the cumulative LLR exceeds
$\log\tfrac{1-\beta}{\alpha}$ and reject below $\log\tfrac{\beta}{1-\alpha}$, with
$\alpha{=}10^{-4}$ and $\beta{=}10^{-2}$, yielding a time-to-authenticate (TTA).
Each window is matched against epochs $\{e{-}1,e,e{+}1\}$ to tolerate
desynchronization, and per-epoch renewal with epoch/nonce binding provides replay
resistance. As in localized audio watermarking~\cite{sanroman2024audioseal}, this
per-window design lets the verifier authenticate as soon as enough evidence accrues, over $2$\,s
windows as the stream arrives. Generation is not streaming in the same sense: writing the message into
the latent requires encoding a whole epoch, so the prover emits nothing until the window is complete.

\section{Security Analysis}\label{sec:security}
Let $\mathrm{Gen}(S,\key,e,r)$ be watermarked motion for content $S$ at epoch $e$ under nonce $r$, and
$\mathrm{Base}(S)$ the watermark-free baseline: the same reconstruction with a fresh
$\zp\sim\Norm(0,I_\dq)$ in place of the keyed draw (identical $\zm$). The two differ only
in the style latent, so undetectability is between a signer and a non-signer; any upstream
anonymization acts on $\zm$ and is common to both.

\begin{definition}[WM-IND]
Following the undetectability notion in~\cite{christ2024undetectable}, a
scheme is \emph{watermark-indistinguishable} if for every probabilistic
polynomial-time distinguisher $\mathcal{D}$ and every $N$ polynomial in the security
parameter, the advantage on a bag of $N$ samples,
$\big|\Pr[\mathcal{D}(\mathrm{Gen}^{(1)},\dots,\mathrm{Gen}^{(N)})=1]-
\Pr[\mathcal{D}(\mathrm{Base}^{(1)},\dots,\mathrm{Base}^{(N)})=1]\big|$, is negligible;
the detection ROC AUC is then $\tfrac12$ up to that advantage. Quantifying over $N$ is
what separates a renewed payload from a fixed one, which is marginally Gaussian yet
distinguishable from $N{=}2$.
\end{definition}

\begin{proposition}[Undetectability by construction]\label{prop:wmind}
Under the keyed embedding above with an ideal pseudorandom function, the laws of
$\mathrm{Gen}(S,\key,e,r)$ and $\mathrm{Base}(S)$ are identical, and so are their
$N$-fold products for every $N$; hence AUC $=\tfrac12$. With a real keyed PRF, any WM-IND advantage is bounded by the PRF
distinguishing advantage.
\end{proposition}
\begin{proof}[Proof sketch]
With whitening, $b=\mathrm{rep}(c)\oplus\Prf_{\key}$ is pseudorandom, so the signs
$2b-1$ are pseudo-uniform and independent of the magnitudes $|\zeta|$; therefore
$\zp$ is distributed as $\Norm(0,I_\dq)$, identical to the baseline's $\zp$. Both
branches pass the \emph{same} decoder $\dec$ over the \emph{same} $\zm$, so the
output laws coincide and no distinguisher beats chance, for any decoder and any
discriminator, over the distribution the two branches induce. What a finite experiment
can show is bounded by its own sampling error; Section~\ref{sec:eval-wmind} measures an interval
that contains chance. Because the message renews with the epoch counter, this holds jointly across the
many samples a long-lived observer collects, not only per sample; a fixed payload
would be marginally Gaussian yet jointly distinguishable.
Replacing the ideal PRF by a keyed one yields the stated reduction.
The only way to break the equality is a generation-only asymmetric operation, which the construction excludes and which we show empirically restores a detectable footprint.
\end{proof}

\begin{proposition}[$(\efa,Q)$-authentication security]\label{prop:forge}
Let a keyless adversary hold no genuine response for any $(e',r)$ with $e'$ inside
the verifier's $\pm$slack window around $e$, let nonces never repeat, and let the adversary make at
most $Q$ verification attempts, choosing each from the accept/reject transcript of the ones before
it. Write $V(x){=}1$ for the verifier accepting candidate $x$, and
\[\efa=\sup\nolimits_{i,\,h_{i-1}}\Pr\bigl[V(x_i){=}1\mid h_{i-1}\bigr]\]
for the worst-case false-accept probability of a single attempt, over every transcript $h_{i-1}$ the
adversary can reach and every candidate it can form from one. Then its probability of producing
motion the verifier accepts for epoch $e$ is at most $Q\,\efa+\mathrm{Adv}^{\mathrm{prf}}$.
\end{proposition}
\begin{proof}[Proof sketch]
Boole's inequality assumes no independence:
$\Pr[\exists i\le Q: V(x_i){=}1]\le\sum_i\Pr[V(x_i){=}1]$, and each term is at most $\efa$ because
the conditional is, on every history the adversary can reach, so the sum is at most $Q\efa$. The PRF term is the advantage of distinguishing $\Hmac_{\key}$ from random, which also
bounds the advantage of predicting it, and is what keeps $\efa$ from collapsing: without $\key$ the
expected codeword $c^\star=\Ecc(\Hmac_{\key}(\mact))$ is pseudorandom, so a candidate cannot be aimed
at it.

The size of $\efa$ comes from the test's calibration. Under the idealized null the match count is
$\mathrm{Binom}(\nb,\tfrac12)$ and a single window's ${\ge}47/63$ accept has exact tail
$p_{\mathrm{tail}}=5.9\times10^{-5}$, below the design parameter $\alpha{=}10^{-4}$ the test is set
to. The deployed verifier maximizes over the three epoch hypotheses in the $\pm$slack window, so a
union bound over them gives $\efa\le3p_{\mathrm{tail}}\approx1.8\times10^{-4}$, itself below
$3\alpha=3\times10^{-4}$; Section~\ref{sec:eval-forge} measures $1.6\times10^{-4}$ over $516{,}096$
wrong-key windows, and finds the wrong-key match distribution on $\mathrm{Binom}(63,\tfrac12)$
($31.51$ against $31.5$), so the idealization is checked rather than assumed. Instantiating the
bound with $\alpha$ in place of $\efa$ drops the epoch multiplicity and does not survive those
measurements.

We state this as $(\efa,Q)$-security rather than unforgeability because $\efa$ is a deployment
quantity, not a negligible function of a security parameter: at $\nb{=}63$ it does not shrink with
any $\lambda$, so the guarantee is statistical and bounded by the online query budget the threat
model grants. Rate limiting the verifier is therefore load-bearing rather than incidental, and a
deployment that does not bound $Q$ does not inherit the bound.
\end{proof}

\begin{proposition}[Recovery requires the key]\label{prop:keyneed}
Without the keystream $\Prf_{\key}(\mact)$, the combined statistic \eqref{eq:comb} has
zero expected correlation with $c$. More generally, under an ideal PRF the embedded
signs are independent of $c$, so no keyless decoder recovers $c$ better than chance;
with a real PRF the advantage of any such decoder is bounded by the PRF distinguishing
advantage.
\end{proposition}
\begin{proof}[Proof sketch]
The roughly four replicas of bit $c_j$ are embedded as $c_j\oplus\Prf_{\key,i}$
for \emph{distinct} pseudorandom $\Prf_{\key,i}$. Summing them without dewhitening
(Eq.~\eqref{eq:comb} with an absent or wrong keystream) averages signs flipped by
independent pseudorandom bits, so the expectation is $0$. The general statement follows
because an ideal $\Prf_{\key}$ makes the embedded sign pattern a one-time pad on
$\mathrm{rep}(c)$: the pattern is uniform and independent of $c$, so any decoder's
output is too.
\end{proof}

Proposition~\ref{prop:keyneed} is the crux of recapture security: even an adversary that recovers the raw signs perfectly
cannot decode without $\key$. Replay resistance follows from
Proposition~\ref{prop:forge} with per-epoch renewal: a session recorded at epoch
$e$ carries $c^\star(e)$, which mismatches $c^\star(e')$ for every $e'$ outside the
$\pm$slack window the verifier tolerates for clock skew, so a replayed recording
older than that window is rejected exactly as a forgery would be. Inside the window
a same-nonce recording still authenticates by construction, which is the bounded freshness surface we measure and the reason slack is tunable
toward zero.

\section{The Recapture Channel}\label{sec:recapture}
Recapture renders $\xw$, screen-records it, and runs a monocular pose estimator
$\lift$ on a perspective projection $\proj$, recovering joints
$\hat J=\lift(\proj(\rec(\xw)))$. A real estimator projects motion onto the
\emph{clean-motion manifold}, discarding off-manifold detail. Because a
WM-IND-safe watermark is subtle \emph{by design} (its energy sits below the
estimator's reconstruction fidelity), a \emph{generic} estimator strips it to chance, on both the
synthetic and the real recapture channel. We therefore identify an inherent
undetectability/robustness tension, and conclude that recapture, not cropping, is
the watermark's true weak channel.

\paragraph{The receiver-side keyed reader.}
Our key observation is that the verifier holds the key and is part of the system,
so for recapture it may use a \emph{dedicated keyed reader} rather than a generic
estimator. Concretely, we train a $2\mathrm{D}\!\rightarrow\!3\mathrm{D}$ lifter
$\lift_\theta$ together with the extractor $\extr$ on the recaptured distribution
($\proj\!\rightarrow$ lift $\rightarrow$ read), with the prover \emph{frozen} so
that generation-side WM-IND is untouched. The reader is conditioned on the key by
FiLM, and its objective carries an invariance term that drives the readout to chance
when the conditioning is absent or wrong. That term matters for reading the figures
below: it makes the no-key and wrong-key bars a training outcome rather than a
measurement of necessity. This reader recovers the watermark close to the clean read, while
estimator-level watermark indistinguishability still holds: a keyless party running the reader obtains Gen-versus-Base AUC
$\approx 0.5$ at the estimator.

\paragraph{Where the security comes from.}
Our readers are key-conditioned, and removing that conditioning drops them to chance,
but the conditioning is not what makes recovery hard to obtain: a reader trained
without it recovers once it is given the keystream. The gate is the
\emph{same keyed whitening} as everywhere else (Proposition~\ref{prop:keyneed}), and a
\emph{fully keyless} read is exact chance (Figure~\ref{fig:ladder}). A Kerckhoffs adversary that trains its own reader from scratch and is then handed the
dewhitening key still falls short of the deployed one (Figure~\ref{fig:ladder}), which reframes
recovery into a \emph{key-gated} capability of the dedicated verifier rather than a property of any
one network. Section~\ref{sec:eval-recap} reports the numbers, on the synthetic channel and through
the full render-to-video loop.

\begin{figure}[tbp]\centering
\includegraphics[width=0.8\linewidth]{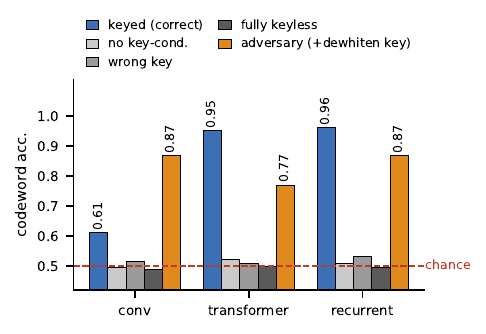}
\caption{Key-necessity ladder. Codeword accuracy for each reader backbone under four key conditions:
correct key, key conditioning removed, wrong key, and fully keyless. The first two sit at chance
because the reader is trained to put them there, so the independent evidence is the orange bar: a
from-scratch Kerckhoffs adversary that never sees the key yet is handed the dewhitening key at test
time.}
\label{fig:ladder}
\end{figure}

\section{Implementation}\label{sec:impl}
We train the motion VAE on HumanML3D~\cite{guo2022humanml3d}, which is built on
AMASS~\cite{mahmood2019amass}. All channels of the threat model are implemented as differentiable layers in an
online curriculum, including a trained monocular lifter in the style of
VideoPose3D~\cite{pavllo2019videopose3d} for recapture. The prover (VAE and codec)
and the verifier (extractor, dewhiten/combine, SPRT) are trained in stages, and
the receiver-side keyed reader is trained with the prover frozen. Because it is
trained against a frozen prover rather than jointly with one, the reader is not tied
to the checkpoint it saw: the shipped reader was trained against an earlier checkpoint
of the prover and reads the shipped prover at $0.949$ on the projected-2D channel and
$0.809$ through the full render-to-RGB loop. The mark is therefore a property of the construction
and not of a co-adapted pair, which is what lets a verifier keep reading after the
generator is updated. We additionally
evaluate on the BOXRR-23 VR dataset~\cite{nair2024boxrr} (three-point head and
hand telemetry), which most closely matches deployed VR sensing; the skeleton and
sensing modality differ from HumanML3D, so that evaluation trains the same
architecture on BOXRR motion.
\smallskip\noindent\textbf{Reproducibility.}
The motion VAE is a six-layer transformer encoder and decoder ($d{=}512$, $8$ heads, feed-forward width $1024$) with a content
latent $\zm\in\R^{(T/4)\times 256}$ downsampled $4\times$ in time and a style latent $\zp\in\R^{256}$ injected by AdaLN; the
codec is $\mathrm{BCH}(63,30,t{=}6)$ keystream-whitened across the $256$ style dimensions with a $10$\,s epoch;
the extractor is an eight-layer transformer ($d{=}384$, $6$ heads, FFN $1536$) over a one-level Haar DWT. We
train with AdamW (learning rate $1$ to $5\times10^{-4}$, weight decay $10^{-4}$, gradient clip $1.0$, batch
$64$ to $128$, seed $42$); the Kullback--Leibler term uses $0.1$-nat free bits and $\beta{=}0.02$ over a $20$k-step warmup. On a
single NVIDIA A100 the VAE trains in $\approx\!4$\,h ($120$k steps), the watermark end to end in $\approx\!8.5$\,h
($150$k steps), and the receiver-side keyed reader in $\approx\!5.5$\,h ($60$k steps), about $18$ GPU-hours
for the released model plus the auxiliary 2D-to-3D lifters.

\section{Evaluation}\label{sec:eval}
The evaluation answers five questions in order, one per subsection, before comparing against
baselines and turning to adaptive attacks. Does MoSign
authenticate robustly, and does it survive the benign channels of the threat model? Is it
undetectable to a keyless observer? Does the key holder still recover the mark through
recapture, and is that recovery genuinely gated on the key? Is the scheme replay-resistant, and
how does forgery scale with the adversary's query budget? And does it compose with an anonymizer without leaking identity? The first
bears on goal~(iv); the rest bear on C3, C4, C2 and C3, as their subsection titles record. C1 is a claim about where the mark is
written and what it is for, so no single experiment tests it; it rests on the authentication,
recapture, and baseline results together. Codeword accuracy is the dewhitened,
soft-combined $63$-bit accuracy. The true-positive rate at a fixed false-positive rate (TPR at FPR) uses the analytic binomial null, checked
against $1400$ empirical null windows for the single-window columns and $10^4$ null
trials per run for the sequential test.

\subsection{Authentication accuracy and robustness}\label{sec:eval-auth}
Table~\ref{tab:robust} reports per-channel codeword accuracy and detection.
MoSign authenticates at \textbf{TPR $=1.0$ at FPR $=10^{-4}$} on clean motion and on
resampling, smoothing, noise and retargeting, with an empirical FPR of $0$ and a median TTA of
$2$ seconds (Figure~\ref{fig:tta}). Two benign channels fall short of that: dropout at $0.96$, and
cropping, the weakest at $0.78$, where a short surviving segment correlates the replica errors; even
there the temporal-DWT read-out keeps the codeword accuracy at $0.87$. A carrier with bandwidth far below an image's therefore
authenticates near-perfectly and degrades gracefully. How far the construction reaches past
this configuration, to unseen motion distributions, to an external frozen generator, and to
other payload lengths, is reported in Appendix~\ref{app:general}.

\begin{table}[tbp]\centering\footnotesize
\caption{Per-channel authentication. \emph{acc.}\ and \emph{TPR} are single-window decisions:
one read of a $196$-frame window against a fixed $47/63$ threshold ($\alpha{=}10^{-4}$), empirical
FPR $0/1400$. \emph{SPRT} is the deployed verifier: $2$\,s windows at $0.5$\,s stride, $256$ streams
per channel over three session nonces (standard deviation ${\le}0.013$), $p_1$ calibrated on the same
$2$\,s statistic, false accepts $0/30000$ each against unwatermarked motion and a wrong key.
\emph{recapture} is the differentiable oracle channel; a real pose estimator is
Table~\ref{tab:recap}. \emph{all seven} is a random ordered composition of the seven single
attacks, each at full strength.}\label{tab:robust}
\begin{tabular}{lccc}
\toprule
Channel & acc. & TPR@$10^{-4}$ & SPRT\\
\midrule
clean      & 0.98 & 1.00 & 0.99\\
resample   & 0.98 & 1.00 & 0.95\\
smooth     & 0.98 & 1.00 & 0.99\\
noise      & 0.98 & 1.00 & 0.99\\
retarget   & 0.97 & 1.00 & 0.99\\
recapture (oracle) & 0.97 & 1.00 & 0.99\\
all seven & 0.96 & 0.99 & 0.92\\
crop       & 0.87 & 0.78 & 0.78\\
dropout    & 0.89 & 0.96 & 0.95\\
\bottomrule
\end{tabular}

\end{table}

\begin{figure}[tbp]\centering
\includegraphics[width=0.8\linewidth]{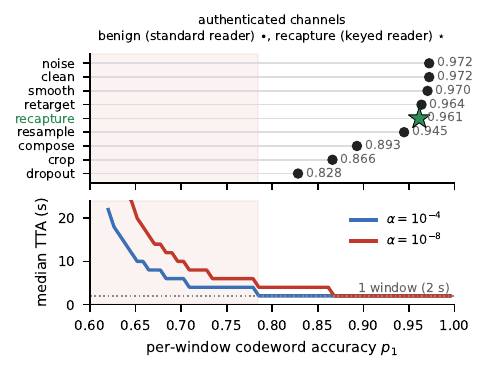}
\caption{Time-to-authenticate. \emph{Top:} per-window codeword accuracy $p_1$: eight benign channels
via the standard extractor (dots), the synthetic recapture channel via the keyed reader (star).
\emph{Bottom:} SPRT time-to-authenticate against $p_1$ (independent $2$\,s windows, $63$-bit
codeword) at two false-accept levels. The shaded band is $p_1$ below the one-epoch threshold; every
channel lies to its right, so all accept on the first window at $\alpha{=}10^{-4}$. At
$\alpha{=}10^{-8}$ the threshold rises to $0.864$ and dropout ($0.828$) needs two windows.}
\label{fig:tta}
\end{figure}

\subsection{Undetectability (C3)}\label{sec:eval-wmind}
Against a worst-case discriminator zoo (a neural detector isomorphic to the
extractor, a permutation two-sample test on the radial-basis-function maximum mean
discrepancy (RBF-MMD), and per-dimension Kolmogorov--Smirnov tests), the
\emph{latent} distribution is preserved. Tested directly on the deployed $\zp$
($\delta{=}0$, fresh session nonce, $N{=}4096$), a trained latent discriminator
measures AUC $0.50$, the pooled KS statistic is $0.0009$ with one of the $256$
per-dimension statistics just outside the family-wise null band ($0.029$ against
$0.025$), and the RBF-MMD permutation test cannot separate it from a fresh
$\Norm(0,I)$ draw ($p=0.86$), as Proposition~\ref{prop:wmind} requires, since
whitening keeps the marginal $\zp$ at $\Norm(0,I)$ and any latent test is then exactly
chance. The neural discriminator, which reads the rendered \emph{motion}
rather than the latent, measures \textbf{AUC $=0.51$} (95\% CI $[0.49,0.54]$,
$N{=}2048$), an interval that contains chance
(Figure~\ref{fig:wmindroc}, blue). Passive indistinguishability is not the same
property as robustness to an adversary that modifies the motion: a gradient can
degrade the band the mark occupies without separating the two branches.
Table~\ref{tab:wmind} and Figure~\ref{fig:wmindroc} substantiate the two requirements of the embedding: removing the keystream whitening (AUC $0.99$) or a single
generation-only asymmetric operation (AUC $0.78$ for $\Delta\theta$, up toward $1$ for
an imprint objective) makes the watermark \emph{strongly} detectable, so the symmetric
whitened construction is a \emph{necessary} condition and prior asymmetric designs would
be trivially detectable in our worst-case setting.

\begin{table}[tbp]\centering\footnotesize\caption{WM-IND and the symmetry ablation. \emph{Worst-case} is the strongest of the
discriminators above ($N{=}2048$ in the motion domain, $4096$ in the latent).}\label{tab:wmind}
\setlength{\tabcolsep}{4pt}
\begin{tabular}{lc}
\toprule
Construction & worst-case AUC\\
\midrule
MoSign (symmetric, final)                 & \textbf{0.51}\\
\;+ amplitude correction $\Delta\theta$ & $\approx 0.78$\\
\;+ imprint objective & $\approx 1.00$\\
\bottomrule
\end{tabular}

\end{table}

\begin{figure}[tbp]\centering
\includegraphics[width=0.8\linewidth]{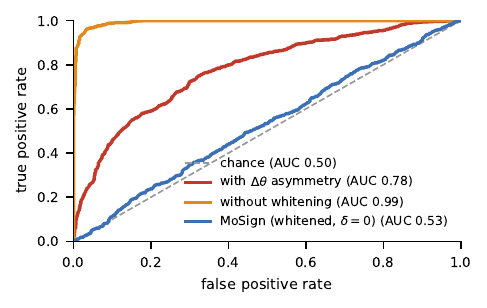}
\caption{WM-IND ROC of the worst-case discriminator (C3), one independently trained discriminator
per curve. The whitened construction (blue) holds the chance diagonal across repeated trainings
(AUC $0.50$ to $0.53$); removing the keystream whitening (orange) or adding a generation-only
$\Delta\theta$ asymmetry (red) lifts it off.}
\label{fig:wmindroc}
\end{figure}

\subsection{Recapture and key-necessity (C4)}\label{sec:eval-recap}
Recapture, the defining VR channel, renders the watermarked motion, screen-records it,
and re-estimates pose with a monocular estimator. We measure it in two settings, a
synthetic projected-2D proxy and the full render-record-reestimate loop (sim-to-real,
below), against three readers of the recaptured stream (Table~\ref{tab:recap}): a
\emph{generic estimator} that holds no key (the attacker), a \emph{generation-side}
watermark trained to survive such estimators, and the \emph{receiver-side keyed reader}
that holds the verifier's key (our contribution).

On the synthetic channel the generic estimator strips the mark to chance ($0.52$). The
generation-side route recovers $0.74$ against the in-distribution estimator it trained
against, but only $0.64$ against a held-out estimator of different architecture and data
split, the architecture-general ceiling the undetectability/robustness tension predicts.
The receiver-side keyed reader instead recovers $0.95$ to $0.96$ on its two stronger
backbones, essentially the clean read.

Figure~\ref{fig:ladder} isolates what gates recovery. Every reader degrades to chance once the key is
withheld: dropping the in-estimator key conditioning (no-key), supplying a wrong
key, or going \emph{fully keyless} (no key conditioning and no dewhitening key)
all sit on the chance line ($0.49$ to $0.53$ across the nine cells), as Proposition~\ref{prop:keyneed} predicts, and the
estimator-level WM-IND AUC stays $\approx 0.50$ throughout (so the readers are
recoverers, not detectors). Three architecturally independent backbones (a dilated
convolution, a transformer, and a recurrent network) all read the same frozen
watermark above chance, two of them well, so the recoverable signal lives in the motion
and key, not in a single
network. Finally, the strongest keyless adversary we can construct, a reader
trained from scratch and then \emph{handed} the dewhitening key, still only
reaches $0.77$ to $0.87$ (and at conv even exceeds that backbone's keyed reader,
showing the in-estimator conditioning is a strength knob, not the security gate),
yet collapses to chance without the key it was handed.

\begin{table}[tbp]\centering\footnotesize\caption{Recapture. \emph{synth.} is the projected-2D channel; \emph{real loop} is render to RGB,
off-the-shelf 2D detection, and a real lift, with no scale oracle. Generic estimators strip the mark
on both channels; the key holder recovers it on both.}\label{tab:recap}
\setlength{\tabcolsep}{4pt}
\begin{tabular}{lcc}
\toprule
Reader of the recaptured motion & synth. & real loop\\
\midrule
generic estimator (no key)               & 0.52 & 0.50\\
\midrule
generation-side, in-distribution & 0.74 & --\\
generation-side, held-out & 0.64 & 0.52\\
\midrule
receiver-side keyed (conv) & 0.61 & --\\
receiver-side keyed (transformer) & 0.95 & --\\
receiver-side keyed (recurrent) & 0.96 & \textbf{0.81}\\
\;\;fully keyless (any backbone) & $\approx 0.50$ & 0.49\\
\bottomrule
\end{tabular}

\end{table}

\paragraph{Sim-to-real: the real recapture loop.}
The synthetic projected-2D channel abstracts the render-then-re-estimate loop; we also run it in full
(Figure~\ref{fig:sim2real}): we fit the watermarked motion to a textured body model~\cite{loper2015smpl}, render
it to RGB, detect 2D keypoints with an off-the-shelf detector~\cite{jocher2023yolov8}, and lift them to 3D with a
real estimator pretrained on motion capture~\cite{pavllo2019videopose3d}. The three-reader hierarchy persists on
this loop (Table~\ref{tab:recap}, right column): the generic estimator strips the mark to chance
($0.50$), the generation-side route that reached $0.64$ on a held-out \emph{synthetic} estimator also falls to
near chance ($0.52$) against a \emph{real} one (so the cross-architecture proxy mildly overestimates it), while
the key-holding verifier's keyed reader, applied \emph{zero-shot} (trained only on the synthetic channel),
recovers far above both floors. Figure~\ref{fig:g8recap} tests it under deployment
conditions: with \emph{no} scale oracle the key-holder holds a flat ${\approx}0.81$ from a frontal
view to a $90^\circ$ side profile, while the keyless floor stays at chance at every viewpoint
($\mathrm{TPR}@10^{-4}{\leq}0.004$, zero in $63$ of $64$ cells; YOLO detection ${\geq}0.996$ throughout). The per-clip scale oracle a deployed
verifier lacks is not load-bearing, moving recovery by only $0.001$ ($0.818\!\to\!0.817$); a paired
perfect-2D projection control (the reader's idealized channel) holds at $0.95$ across the same viewpoints, so the
watermark is \emph{not} viewpoint-sensitive and the constant ${\approx}0.14$ gap is the cost of the real
render-and-detect pipeline, not the mark. Joint-dropout occlusion degrades recovery gracefully
($0.81\!\to\!0.78$ at $30\%$ of joints) while the keyless party never authenticates. The recapture robustness
thus holds \emph{without an oracle and across the viewpoints} a real screen-record imposes; its operating conditions and limits are taken up in the discussion.

\begin{figure*}[tp]\centering
\begin{tikzpicture}[mosfig]
  \node[plain, text width=24mm, minimum height=13mm] (m) at (0,0) {};
  \node[font=\scriptsize, anchor=north] at (0,0.523) {watermarked motion};
  \foreach \d/\k in {-0.90/poseA, -0.30/poseB, 0.30/poseC, 0.90/poseD}{%
    \pic[draw=mosNeu!85] at (\d,-0.30) {\k};}
  \cardn{r}{mosNeu}{24mm}{(3.07,0)}{11mm}{\icw{screen}render}{fit a body model,\\rasterize to RGB}
  \cardn{d}{mosAcc}{24mm}{(6.12,0)}{11mm}{\icw{cam}record}{off-the-shelf 2D\\keypoint detector}
  \cardn{l}{mosKey}{24mm}{(9.17,0)}{11mm}{\icw{keyicon}re-estimate}{lift 2D to 3D:\\generic or keyed}
  \cardn{x}{mosVer}{19mm}{(11.98,0)}{11mm}{\icw{check}read}{extractor\\and SPRT}
  \node[card={mosKey}{22mm}, minimum height=7mm] (key) at (9.17,-1.34) {verifier key $\key$};

  \node[card={mosAdv}{22mm}, minimum height=12.5mm] (o1) at (14.93, 0.72) {};
  \node[font=\scriptsize, align=center, text width=22mm, anchor=north, inner sep=0pt]
       at ([yshift=-1.4mm]o1.north) {generic lifter:\\chance, $0.50$};
  \node[card={mosVer}{22mm}, minimum height=12.5mm] (o2) at (14.93,-0.72) {};
  \node[font=\scriptsize, align=center, text width=22mm, anchor=north, inner sep=0pt]
       at ([yshift=-1.4mm]o2.north) {key-conditioned:\\\textbf{0.81}};
  \foreach \n/\c/\v in {o1/mosAdv/0.00, o2/mosVer/0.67}{%
    \begin{scope}[shift={(\n.south)}]
      \fill[\c!12] (-0.95,0.13) rectangle (0.95,0.30);
      \fill[\c!55] (-0.95,0.13) rectangle (\v,0.30);
      \draw[\c!45, line width=0.4pt] (-0.95,0.13) rectangle (0.95,0.30);
      \draw[mosNeu!60, densely dashed, line width=0.4pt] (0,0.09) -- (0,0.34);
    \end{scope}}

  \begin{scope}[on background layer]
    \node[chan, fit=(r)(d)] (PX) {};
  \end{scope}
  \node[lab, above=0.4mm of PX.north] {the signal is pixels here};
  \draw[flow] (m) -- (r); \draw[flow] (r) -- (d); \draw[flow] (d) -- (l); \draw[flow] (l) -- (x);
  \draw[keyflow] (key) -- (l);
  \coordinate (fk) at (13.26,0);
  \draw[flow, -] (x.east) -- (fk);
  \draw[flow] (fk) |- (o1.west);
  \draw[flow] (fk) |- (o2.west);

  \node[draw=mosNeu!35, fill=mosNeu!5, rounded corners=2pt, line width=0.4pt,
        minimum width=15mm, minimum height=10mm] (T1) at ( 1.55,-2.45) {};
  \node[draw=mosAcc!40, fill=mosAcc!6, rounded corners=2pt, line width=0.4pt,
        minimum width=15mm, minimum height=10mm] (T2) at ( 4.60,-2.45) {};
  \node[draw=mosAcc!40, fill=mosAcc!6, rounded corners=2pt, line width=0.4pt,
        minimum width=15mm, minimum height=10mm] (T3) at ( 7.65,-2.45) {};
  \node[draw=mosNeu!35, fill=mosNeu!5, rounded corners=2pt, line width=0.4pt,
        minimum width=15mm, minimum height=10mm] (T4) at (10.70,-2.45) {};
  \pic at ( 1.55,-2.45) {rep3d};  \pic at ( 4.60,-2.45) {reppix};
  \pic at ( 7.65,-2.45) {rep2d};  \pic at (10.70,-2.45) {rep3dr};
  \foreach \xg in {1.55, 4.60, 7.65, 10.70}{%
    \draw[mosNeu!35, densely dotted, line width=0.4pt] (\xg,-0.80) -- (\xg,-1.95);}
  \node[lab, anchor=north, align=center, text width=22mm] at ( 1.55,-3.02) {3D skeleton};
  \node[lab, anchor=north, align=center, text width=22mm] at ( 4.60,-3.02) {rendered pixels};
  \node[lab, anchor=north, align=center, text width=26mm] at ( 7.65,-3.02) {2D keypoints};
  \node[lab, anchor=north, align=center, text width=22mm] at (10.70,-3.02) {3D skeleton,\\re-estimated};
  \node[lab, anchor=west] at (12.60,-2.45) {what travels on each arrow};
\end{tikzpicture}
\caption{The real recapture loop: watermarked motion fitted to a body model, rendered to RGB, 2D
keypoints detected, then lifted back to 3D. Between render and lift the signal is not a skeleton at
all, so nothing guarantees a subtle latent perturbation survives. It does, but only for the key
holder.}
\label{fig:sim2real}
\end{figure*}
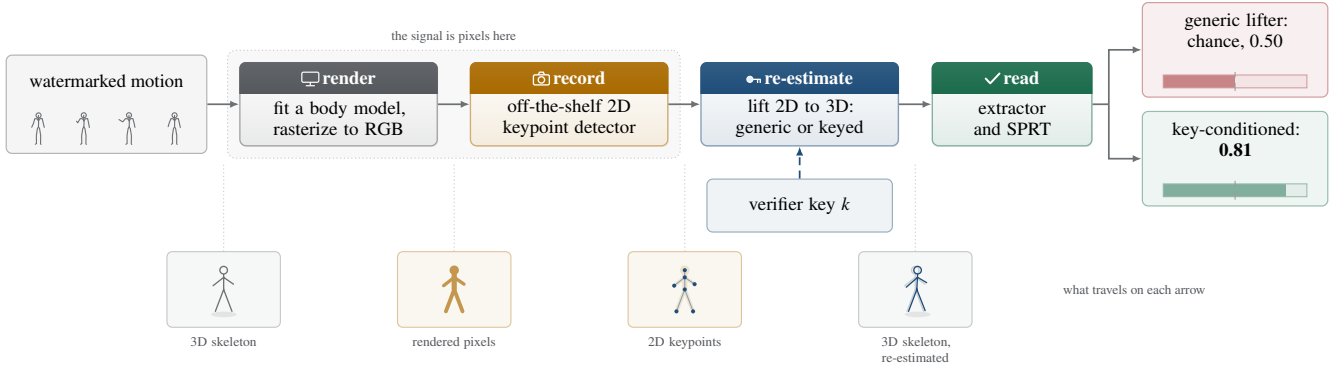

\begin{figure}[tbp]\centering
  \includegraphics[width=0.8\columnwidth]{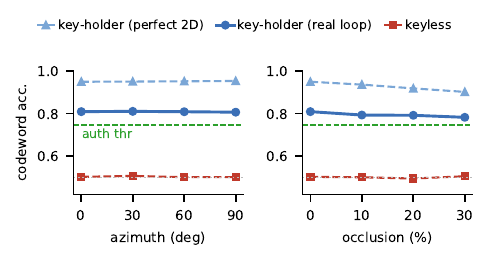}
  \caption{Hardening the recapture claim (real render-to-detection loop, $N{=}288$, no scale oracle).
  \emph{Left:} key-holder recovery across azimuth, against a paired perfect-2D projection control.
  \emph{Right:} the same loop under joint-dropout occlusion. The keyless floor stays at chance
  throughout.}
  \label{fig:g8recap}
\end{figure}

\subsection{Unforgeability (C2)}\label{sec:eval-forge}
We measure whether a keyless adversary can make the verifier \emph{accept} forged
motion, the empirical counterpart of Proposition~\ref{prop:forge}, benchmarked
against the black-box watermark-forgery attack of~\cite{muller2025blackbox}. Three
adversaries observe signed motion in public VR but never hold the key: a
\emph{replay} adversary that resubmits an observed signature under the verifier's
fresh nonce; a \emph{gradient} adversary that fully controls the white-box
extractor output and crafts motion toward an arbitrary target codeword; and a
black-box \emph{query} adversary that submits up to $Q{=}64$ perturbed candidates
and keeps any the verifier accepts. Figure~\ref{fig:forgery} reports forgery
success against the SPRT design floor $\alpha=10^{-4}$. The two offline adversaries sit
at the floor and the query adversary does not: replay at $10^{-4}$ ($1/10056$ trials, 95\% Wilson upper bound $4.5\times10^{-4}$);
the gradient adversary at $5\times10^{-4}$ \emph{despite} matching its own chosen
target codeword $81\%$ of the time, which shows that controlling the extractor does
not help when the target is the secret keyed codeword; and the query adversary reaches
$1.7\times10^{-2}$ at $Q{=}64$ ($4144$ trials), rising about linearly in the budget from
$2.4\times10^{-4}$ at $Q{=}1$. That $Q{=}1$ figure estimates the $\efa$ of
Proposition~\ref{prop:forge} for this adversary, and $Q\hat\efa$ tracks the sweep:
$2.4\times10^{-4}$, $1.9\times10^{-3}$, $7.7\times10^{-3}$ and $1.5\times10^{-2}$ at
$Q{=}1,8,32,64$ against a measured $2.4\times10^{-4}$, $2.4\times10^{-3}$, $1.0\times10^{-2}$ and
$1.7\times10^{-2}$, a ratio of $1.00$ to $1.34$. $Q\hat\efa$ is not itself a bound: $\hat\efa$ is a
point estimate from finitely many trials, and the ratio exceeding one reflects that rather than any
failure of the union bound, which needs no independence among the queries. The one-sided $95\%$
upper confidence bound on $\efa$, $1.1\times10^{-3}$ from $1/4144$, puts every measured rate below
$Q\efa$. What the dependence among queries does rule out is reading the sweep as
$1-(1-\efa)^Q$. This adversary is therefore bounded by the query budget the threat model grants it
rather than by the mark. Freezing both the message and the keystream, so the embedded
sign pattern never renews, is forged $100\%$ of the time by both replay and a single query.
Separating the two shows which one carries the freshness: with a \emph{fixed} message but a
challenge-bound keystream, replay still fails at $9.9\times10^{-5}$, exactly the rate of the full
scheme, because $\mathrm{rep}(c)\oplus\Prf_{\key}(\mact)$ is re-randomized by the pad alone. It is
the challenge-bound keystream, not renewal of the message, that leaves the adversary no fixed
target; the message binds the response to the challenge semantically and drives the epoch
matching, but it is not what defeats replay.

\begin{figure}[tbp]\centering
\includegraphics[width=0.8\linewidth]{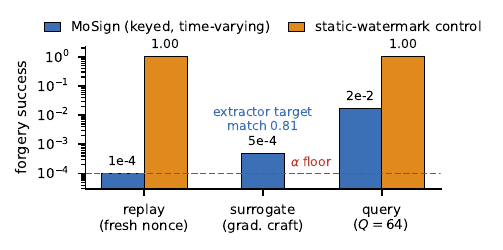}
\caption{Forgery benchmark (C2). Forgery success, meaning the verifier accepts, on a log scale for
three keyless adversaries and a non-renewing control, against the SPRT design floor $\alpha$
(dashed).}
\label{fig:forgery}
\end{figure}

\paragraph{Cross-key false-accept at population scale.}
The forgery benchmark crafts a signature; the dual question is whether a verifier holding key $A$ ever accepts
a different legitimate signer using key $B$. We enroll $K{=}64$ keys, sign every motion instance with each key,
and read it back with every verifier key (the deployed $\pm1$-epoch max). Because a wrong key un-whitens to an
independent codeword, the cross-key codeword match is a fair coin (Figure~\ref{fig:crosskey}): its mean is
$31.5$ of $63$, matching $\mathrm{Binom}(63,\tfrac12)$ ($31.51$ measured), with no partial leakage. A single
window's ${\ge}47/63$ accept sits at the binomial tail ($P{=}5.9\times10^{-5}$, or $1.8\times10^{-4}$ under the
three-epoch maximum), and the measured pooled cross-key false-accept is $1.6\times10^{-4}$ ($84/516{,}096$),
not growing with $K$ ($1.4$ to $2.6\times10^{-4}$ across $K\in\{8,16,32,64\}$, no monotone multiplication). The deployed decision is
stronger still: running the actual SPRT (the $\pm1$-epoch max with its $\log 3$ correction) over $10^6$ impostor
streams accepts \emph{none} (operating false-accept ${\le}3.8\times10^{-6}$, 95\% Wilson), so the headline
$10^{-4}$ floor holds across keys and after the epoch multiplicity, not only under the single-key binomial null.

\begin{figure}[tbp]\centering
\includegraphics[width=0.8\linewidth]{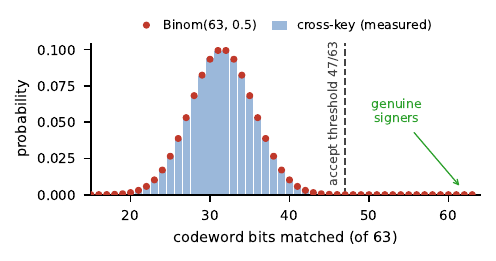}
\caption{Cross-key codeword-match distribution (C2), $K{=}64$ keys. Under the wrong key the match
(bars) lands on $\mathrm{Binom}(63,\frac12)$ (markers), while genuine signers recover
${\approx}61.7/63$. The $10^6$-stream SPRT result treats the $516{,}096$ measured windows as
independent; windows within one real stream are correlated, so it is a bound under that assumption.}
\label{fig:crosskey}
\end{figure}

\paragraph{Epoch-slack, nonce misuse, and splicing.}
Anti-replay rests on the epoch window and nonce freshness, which we sweep rather than assert
(Figure~\ref{fig:g10relay}). The epoch-slack is one knob with two faces: the clock-skew / transmission
tolerance and, equally, the relay window. A same-nonce response replayed $d$ epochs late is accepted exactly
when $d\le\mathrm{slack}$, so a forward relay can be at most $(\mathrm{slack}{+}1)\tau$ stale (${\le}20$\,s at
the deployed slack $1$, $\tau{=}10$\,s) before the slack rejects it; \emph{distance-bounding} is the standard
liveness mitigation. A cross-session replay under a fresh verifier-issued nonce is rejected at every slack
($0/512$), while reusing one $(\text{epoch},\text{nonce})$ revives replay ($1.0$ accept); nonce hygiene is a
stated assumption. The strong attack on the windowed SPRT is \emph{splicing}: an adversary concatenates
recorded high-signal signed windows against the live nonce. Because each window's codeword is bound to its old
$(\text{epoch},\text{nonce})$, against the live nonce its match is a fair coin ($31.3$ versus
$\mathrm{Binom}(63,\tfrac12){=}31.5$), so splicing real windows beats a random impostor by nothing and
best-of-$Q$ accepts none even at $Q{=}10^4$; the residual brute-force ceiling is closed by rate-limiting, not
by a watermark property.

\begin{figure}[tbp]\centering
\includegraphics[width=0.8\linewidth]{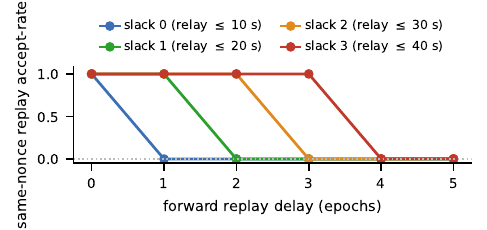}
\caption{Epoch-slack tradeoff (C2). Same-nonce delayed-replay accept rate versus forward replay
delay, one curve per epoch-slack. Slack is at once the clock-skew / relay tolerance and the replay window.}
\label{fig:g10relay}
\end{figure}

\subsection{Composition with anonymization (C3)}\label{sec:eval-anon}
\begin{table}[tbp]\centering\footnotesize
\caption{Composition with anonymization (C3): rank-1 re-identification over $1000$ BOXRR-23 users
under Deep Motion Masking's own note-anchored protocol.}\label{tab:composition}
\begin{tabular}{lc}
\toprule
Condition & re-id (rank-1)\\
\midrule
clean                & 0.884\\
anonymized (DMM)     & 0.034\\
VAE recon            & 0.026\\
watermarked (MoSign) & 0.024\\
\midrule
chance               & 0.001\\
\bottomrule
\end{tabular}

\end{table}
MoSign is designed to stack on top of an upstream anonymizer; we verify it adds
\emph{no} de-anonymization side channel when composed with the state-of-the-art VR
anonymizer, Deep Motion Masking~\cite{nair2024deepmasking}. The user's motion is
anonymized, and MoSign then watermarks the anonymized motion (encode, overwrite $\zp$
with the keyed draw, decode). We measure identity leakage under the anonymizer's own
published re-identification protocol (a note-anchored attacker over $1000$ BOXRR-23
users), which reduces re-identification from $0.884$ on clean motion to $0.034$
anonymized (chance $0.001$; Table~\ref{tab:composition}).

Turning the watermark on does not make the user more identifiable. Applied to the
anonymized motion, MoSign's watermark survives for the key holder (codeword accuracy $0.987$) while leaving re-identification unchanged:
the watermarked motion re-identifies at $0.024$, statistically indistinguishable from
the anonymized-only motion ($0.034$) and from the watermark-free VAE round-trip
($0.026$, the lossy reconstruction that the watermark rides on). The residual leakage
is bounded by the anonymizer, not the watermark, because the message lives only in the
keyed $\zp$, which is overwritten by a draw that carries neither the user's
style-identity nor a privacy leak. The authentication channel and the privacy of the
motion are thus orthogonal.

\subsection{Baseline comparison}\label{sec:eval-baselines}
We score MoSign, two ablations of its own freshness, and three representative baselines on one
harness, every scheme read by the same code on the same validation split. \emph{Gen-Static} is an
in-generation generative watermark with a fixed payload, a port of Gaussian
Shading~\cite{yang2024gaussianshading} and Stable
Signature~\cite{fernandez2023stablesignature} that keeps our exact embedding but never
renews the codeword. \emph{Post-hoc SS} is a classical additive spread-spectrum
watermark~\cite{cox1997spread} applied post-hoc to the \emph{rendered} motion and read by a
matched filter, in a keyed (time-varying) and a static variant. The generative schemes
share the \emph{same} trained model and differ only in the codec, so any gap is the
paradigm and not capacity. Table~\ref{tab:baselines} reports in-generation, worst-case WM-IND AUC
under a pooled test and under the replica-pair test, replay-forgery rate, codeword accuracy,
recapture recovery, and the watermark footprint.

Freezing the message alone changes nothing measurable, within $0.01$ of MoSign on every
column; freezing the keystream alone is caught, but only by the pair test, which is why that column
is in the table; freezing both is Gen-Static. The post-hoc schemes fail before any of that, on
in-generation and on payload.

\begin{table}[tp]\centering\footnotesize
\caption{Baseline comparison. The first three rows vary where the freshness in $\mathrm{rep}(c)\oplus\Prf_{\key}(\mact)$ comes from; the rest is
prior work. \emph{WM-IND} is the stronger of two pooled tests and \emph{pair} the replica-pair product
of Figure~\ref{fig:multisample}, both worst case over $N\le256$. \emph{Replay} is bit agreement between
a replayed clip and a live challenge, not an accept rate, so it targets $0.5$ (boldface where met).
\emph{Recapture} is read by a held-out keyless estimator. The three baselines are defined above.}\label{tab:baselines}
\setlength{\tabcolsep}{3pt}\scriptsize
\begin{tabular}{lcccccc}
\toprule
Scheme & WM-IND & pair & Replay & Codeword & Recapture & MPJPE (m)\\
\midrule
\multicolumn{7}{@{}l}{\emph{in-generation}}\\
MoSign & \textbf{0.52} & \textbf{0.51} & \textbf{0.49} & 0.97 & 0.62 & 0.044\\
\ \ static message & \textbf{0.52} & \textbf{0.50} & \textbf{0.51} & 0.98 & 0.63 & 0.044\\
\ \ static keystream & \textbf{0.52} & 1.00 & \textbf{0.48} & 0.98 & 0.64 & 0.044\\
\addlinespace[1.5pt]
Gen-Static & 1.00 & 1.00 & 0.97 & 0.97 & 0.57 & 0.043\\
Stable-Sig static & 1.00 & -- & 0.98 & 0.98 & 0.63 & 0.044\\
\addlinespace[2.5pt]
\multicolumn{7}{@{}l}{\emph{post-hoc spread-spectrum}}\\
\ \ keyed & 1.00 & -- & \textbf{0.50} & 0.65 & 0.54 & 0.084\\
\ \ static & 1.00 & -- & 0.60 & 0.60 & 0.51 & 0.085\\
\bottomrule
\end{tabular}

\end{table}

\begin{figure}[tbp]
  \centering
  \includegraphics[width=0.9\columnwidth]{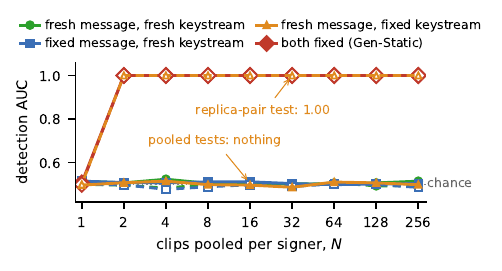}
  \caption{Where the freshness has to come from, on the true codec ($N$ clips pooled per long-lived
signer). Message and keystream are varied independently; the magnitudes stay fresh in all four
combinations. Solid is the stronger of two pooled tests over the bag mean (the $L_2$ norm and a learned mean-pool network); dashed is the replica-pair product
$s_i s_{i+\nb}$. Three curves coincide at $1.00$: the fixed-keystream arm's pair test,
and both tests on the both-fixed arm.}
  \label{fig:multisample}
\end{figure}

\paragraph{A static payload breaks undetectability and unforgeability together.}
MoSign's embedding \emph{is} Gaussian Shading (step~4 of the embedding); Gen-Static keeps
that exact distribution-preserving construction but fixes the payload. By the
Gaussian-Shading theorem a single Gen-Static latent is marginally Gaussian, so it is
\emph{single-sample} undetectable. Our WM-IND adversary is stronger: a keyless party observes a user
across many epochs and tests the joint distribution. A fixed payload pins the sign of every latent coordinate, so across
samples each coordinate carries a constant bias $\mathbb{E}[\zp[i]]=\pm\sqrt{2/\pi}$ that a
trained discriminator detects (AUC $1.00$), and the same fixed codeword replays (forgery
$0.97$). Pooling $N$ clips from one signer traces how this detectability emerges with the observer's
budget (Figure~\ref{fig:multisample}): a keyless test of the per-coordinate mean (the sufficient statistic for an unknown constant bias)
sits at chance at $N{=}1$ but separates the fixed
codeword from anonymized motion at AUC $1.0$ once $N{\ge}2$, whereas the deployed per-epoch renewal stays at
chance (AUC ${\approx}0.50$, at most $0.52$ across all $N$ and both detectors) for every $N$ up to $256$.
What multi-sample undetectability needs is that the embedded sign pattern be fresh across samples,
and either half of $\mathrm{rep}(c)\oplus\Prf_{\key}(\mact)$ supplies that: a fixed message under a
challenge-bound keystream stays at chance for every $N$ we test. The two halves are not
interchangeable, though, and a pooled test cannot see the difference. Because the codeword is replicated across $\zp$, two dimensions carrying the same codeword bit
have signs differing by $(-1)^{w_i\oplus w_{i+\nb}}$, which does not involve the message at all. A
frozen keystream makes that product a constant the observer can average, and the replica-pair test
reaches AUC $1.00$ at $N{\ge}2$ while both pooled detectors sit at chance. Gen-Static freezes both
halves and is caught by every test; whitening under a fixed key injects no cross-sample variation, so
the whitened and non-whitened static ports behave identically (both at AUC $1.00$ in
Table~\ref{tab:baselines}). A static fingerprint would itself link a user across sessions, so
per-epoch renewal is required by the anonymity setting rather than merely convenient.

\paragraph{In-generation buys what the post-hoc baselines cannot.}
A classical post-hoc spread-spectrum mark, added to the rendered motion rather than the
generative latent, cannot carry the keyed payload at MoSign's footprint. Even granted our
keyed nonce, so that it is itself replay-unforgeable (forgery $0.50$), and allowed roughly
twice MoSign's MPJPE and six times its jerk, it reaches only $0.65$ codeword accuracy,
stays detectable (AUC $1.00$, the additive carrier is a learnable subspace), and is stripped
by recapture. This rules out the classical additive family rather than every post-hoc design; a
learned post-hoc encoder might close part of the gap, and we have no lower bound forbidding one. A strength sweep makes the matched-utility point exact: the anonymized host
alone already costs MPJPE $0.044$, MoSign's entire total distortion to three decimals, so
constrained to that budget the mark carries no signal (codeword accuracy $0.50$, chance), and even its faintest readable
setting (codeword accuracy $0.53$) costs $1.6\times$ the MPJPE and $4.3\times$ the jerk. The
budget is spent on the anonymization the host requires, leaving none for a post-hoc mark.
This is the in-generation advantage of C1: writing the message into the generative latent
yields an
undetectable, low-footprint, recapture-survivable payload that a post-hoc additive mark
cannot match.

\subsection{Adaptive attacks and defenses}\label{sec:eval-adaptive}
Reading the watermark needs the key (Figure~\ref{fig:ladder}); this subsection asks whether
replay can defeat the protocol without it. Removal is the other half, and
Appendix~\ref{app:removal} measures every removal attack we implemented.

\paragraph{Renewal defeats replay.}
The codeword and the whitening keystream both derive from
$\Hmac_{\key}(\text{epoch}\,\|\,\text{nonce}\,\|\,\text{ctx})$, so a recording is
valid only for its original $(\text{epoch},\text{nonce})$. Table~\ref{tab:replay}
replays a watermarked clip at increasing epoch offsets: a replay delayed beyond one
epoch, or carried across sessions under a fresh nonce, is rejected $100\%$
(time-to-authenticate infinite), complementing the fresh-nonce replay measured in the
forgery benchmark. The only acceptance window is the
$\pm1$-epoch slack (${\sim}10$\,s of clock-skew tolerance), tunable toward zero at
the cost of boundary-window robustness; real-time relay of a live response is out of
scope, as for any challenge-response.

\begin{table}[tbp]\centering\footnotesize
\caption{Anti-replay ($256$ clips, slack $=1$): a recording is valid only for its
original $(\text{epoch},\text{nonce})$. The accept rates of $1.00$ are $256/256$.}\label{tab:replay}
\begin{tabular}{lcc}
\toprule
Replay scenario & false-accept & TTA\\
\midrule
fresh (legit, $\Delta{=}0$)       & $1.00$          & $2$\,s\\
delayed $+1$ epoch (within slack) & $1.00$          & $2$\,s\\
delayed $+2$ epochs (stale)       & $\mathbf{0.00}$ & $\infty$\\
delayed $+5$ epochs (stale)       & $\mathbf{0.00}$ & $\infty$\\
cross-session (fresh nonce)       & $\mathbf{0.00}$ & $\infty$\\
\bottomrule
\end{tabular}

\end{table}

\section{Discussion and Limitations}\label{sec:disc}
\noindent\textbf{Sim-to-real conditions.}
The real render-to-RGB-to-detection validation (Table~\ref{tab:recap}) reports a
key-holder recovery of $0.81$ with no scale oracle, flat across azimuth and falling to $0.76$ under
elevation and occlusion; harder renders with clothing texture, heavier occlusion, and in-the-wild video
remain future work. For the generation-side route the cross-architecture readers remain our estimator-agnostic
proxy.

\smallskip\noindent\textbf{Scope of the recapture claim.}
Robustness is for the \emph{key-holding} verifier, not ``any estimator''; the
generation-side route is fundamentally bounded by the undetectability/robustness
tension.

\smallskip\noindent\textbf{Utility cost.}
Measured against ground truth the pipeline is far from smooth: jerk is $3.9\times$
($153.3$ against $38.9$\,m/s$^3$) and foot-skating, the fraction of planted-foot frames in which the
foot slides, more than doubles ($0.142$ against $0.066$). Most of that is not the watermark: the same
autoencoder round trip with the mark switched off already reads $138.4$\,m/s$^3$ and $0.126$, so the
round trip accounts for $2.6\times$ of the $3.9\times$ and the mark adds $15$\,m/s$^3$ on top, an
$11\%$ increase, $+0.016$ of foot-skating and $0.027$\,m of positional displacement. Charging the mark
the total would charge it for a cost the round trip incurs on its own and that every scheme we compare
pays, since all are marked on the same reconstructed host. Both are aggregate proxies; whether the
residual is visible at $20$\,fps is a question for a psychophysical study with human subjects.

\smallskip\noindent\textbf{Generation latency.}
Verification is streaming; generation is windowed. The encoder attends over the whole epoch, so at the
shipped $196$-frame window the prover buffers $9.8$\,s at $20$\,fps before emitting, while its compute
is $7.3$\,ms on an L40S and $51.7$\,ms on a CPU, of which the mark itself is $0.6$\,ms. The delay is
specific to embedding in a generative latent: Deep Motion Masking, the anonymizer we compose with, is
causal (a causal convolution over per-frame layers) and adds no lookahead. It is a design point rather
than a floor. Compute is flat in window length and shortening the epoch keeps every window's challenge
fresh, so a $40$-frame epoch buffers $2.0$\,s instead of $9.8$\,s and costs at most $0.01$ of codeword
accuracy on any channel, gaining on cropping (Appendix~\ref{app:general}). We ship the long epoch for
per-read redundancy and this is what it costs; a causal encoder would remove the buffer altogether.

\smallskip\noindent\textbf{Undetectability assumption.}
The guarantee rests on $\zp\sim\Norm(0,I)$; drift to a non-Gaussian $\zp$ weakens
it to $\varepsilon$-undetectability.

\section{Conclusion}\label{sec:concl}
MoSign turns a watermark into an authentication protocol on the motion channel
that is provably undetectable, replay-resistant, forgery-resistant to a bounded-query
adversary, and, uniquely, robust to recapture for the key holder. It is, to our knowledge, the first
in-generation watermark for skeletal motion and the first used for entity
authentication, and it offers a security-grounded answer to the analog hole of
public VR.

\section*{Ethical Considerations}
\noindent\textbf{Human subjects.} This work involved no human subjects. Every measurement runs on
previously collected, publicly released motion, and we collected no new data from people. We ran no user
study; the mark's effect on the motion is measured by MPJPE, jerk and foot-skating.

\smallskip\noindent\textbf{Data and licensing.} HumanML3D~\cite{guo2022humanml3d} and
AMASS~\cite{mahmood2019amass} are public under their respective licenses; BOXRR-23~\cite{nair2024boxrr}
was obtained through its application process and used under its terms; SMPL~\cite{loper2015smpl} is used
under its model license. All three distribute motion under pseudonymous participant identifiers, which we
use only as class labels.

\smallskip\noindent\textbf{Running a de-anonymization attacker.} Section~\ref{sec:eval-anon} runs a
re-identification attack over $1000$ BOXRR-23 users. It is the anonymizer's own published benchmark, run
for one purpose: to test whether adding MoSign opens a de-anonymization side channel. The attacker
separates pseudonymous identifiers inside the dataset, and we made no attempt to link any identifier to a
real person. Both the attacker and the anonymizer are prior work, so the experiment introduces no
capability that did not already exist.

\smallskip\noindent\textbf{Attacks we built.} Evaluating the scheme required implementing
removal, forgery, laundering, and recapture adversaries, some of them granted the extractor weights or the
session key. All of them run against our own system on public data. We attacked no live service, so there
is no third party to notify and no vulnerability to disclose.

\smallskip\noindent\textbf{Dual use.} A mark that a bystander cannot see but a key holder can read is the
kind of primitive that can be turned against the person carrying it. Two properties bound the risk. The mark is readable only with $\key$, so it
creates no public identifier; and it is renewed every epoch under a fresh nonce, so the mark does not
link two sessions for a keyless observer. A verifier holding the persistent shared key can link
them, which is the authentication it was given the key for. A platform that both mandates the watermark and holds the key can
confirm that a stream came from a given key. That is exactly the authentication the scheme provides, which is why
$\key$ belongs to the party the user chooses to prove themselves to and not to the platform by default.

\smallskip\noindent\textbf{Scope.} MoSign composes with an upstream anonymizer and authenticates
through it; removing identity from motion is the anonymizer's job. Relied on as an anonymity
mechanism it would be carrying weight it was never given.


\bibliographystyle{plain}
\bibliography{references}

@inproceedings{nair2023unique,
  title = {Unique Identification of 50,000+ Virtual Reality Users from Head \& Hand Motion Data},
  author = {Nair, Vivek and Guo, Wenbo and Mattern, Justus and Wang, Rui and O'Brien, James F. and Rosenberg, Louis and Song, Dawn},
  booktitle = {32nd USENIX Security Symposium (USENIX Security 23)},
  year = {2023},
  eprint = {2302.08927},
  archivePrefix = {arXiv},
  url = {https://arxiv.org/abs/2302.08927},
}

@misc{nair2023inferring,
  title = {Inferring Private Personal Attributes of Virtual Reality Users from Head and Hand Motion Data},
  author = {Nair, Vivek and Rack, Christian and Guo, Wenbo and Wang, Rui and Li, Shuixian and Huang, Brandon and Cull, Atticus and O'Brien, James F. and Latoschik, Marc and Rosenberg, Louis and Song, Dawn},
  howpublished = {arXiv preprint arXiv:2305.19198},
  year = {2023},
  eprint = {2305.19198},
  archivePrefix = {arXiv},
  url = {https://arxiv.org/abs/2305.19198},
}

@article{schach2025crossxr,
  title = {Motion-Based User Identification across XR and Metaverse Applications by Deep Classification and Similarity Learning},
  author = {Schach, Lukas and Rack, Christian and McMahan, Ryan P. and Latoschik, Marc Erich},
  journal = {Frontiers in Virtual Reality},
  year = {2026},
  note = {arXiv:2509.08539, September 2025},
  eprint = {2509.08539},
  archivePrefix = {arXiv},
  url = {https://arxiv.org/abs/2509.08539},
}

@article{nair2024boxrr,
  title = {Berkeley Open Extended Reality Recordings 2023 (BOXRR-23): 4.7 Million Motion Capture Recordings from 105,852 Extended Reality Device Users},
  author = {Nair, Vivek and Guo, Wenbo and Wang, Rui and O'Brien, James F. and Rosenberg, Louis and Song, Dawn},
  journal = {IEEE Transactions on Visualization and Computer Graphics (TVCG)},
  year = {2024},
  eprint = {2310.00430},
  archivePrefix = {arXiv},
  url = {https://arxiv.org/abs/2310.00430},
}

@inproceedings{nair2024deepmasking,
  title = {Deep Motion Masking for Secure, Usable, and Scalable Real-Time Anonymization of Ecological Virtual Reality Motion Data},
  author = {Nair, Vivek and Guo, Wenbo and O'Brien, James F. and Rosenberg, Louis and Song, Dawn},
  booktitle = {IEEE Conference on Virtual Reality and 3D User Interfaces (IEEE VR)},
  pages = {493--500},
  year = {2024},
  eprint = {2311.05090},
  archivePrefix = {arXiv},
  url = {https://arxiv.org/abs/2311.05090},
}

@article{meng2024avatarhunter,
  title = {De-Anonymizing Avatars in Virtual Reality: Attacks and Countermeasures},
  author = {Meng, Yan and Zhan, Yuxia and Li, Jiachun and Du, Suguo and Zhu, Haojin and Shen, Xuemin (Sherman)},
  journal = {IEEE Transactions on Mobile Computing},
  volume = {23},
  number = {12},
  pages = {13342--13357},
  year = {2024},
  url = {https://ieeexplore.ieee.org/document/10592805/},
}

@inproceedings{agarwal2007tamper,
  title = {Tamper Proofing 3D Motion Data Streams},
  author = {Agarwal, Parag and Prabhakaran, Balakrishnan},
  booktitle = {Advances in Multimedia Modeling (MMM), Lecture Notes in Computer Science vol. 4351, Springer},
  year = {2007},
  url = {https://doi.org/10.1007/978-3-540-69423-6_71},
}

@inproceedings{li2007progressive,
  title = {Watermarking for Progressive Human Motion Animation},
  author = {Li, Shiyu and Okuda, Masahiro},
  booktitle = {IEEE International Conference on Multimedia and Expo (ICME)},
  year = {2007},
  url = {https://doi.org/10.1109/ICME.2007.4284886},
}

@inproceedings{motwani2008skinning,
  title = {Robust Watermarking of 3D Skinning Mesh Animations},
  author = {Motwani, Rakhi C. and Ambardekar, Ameya and Motwani, Mukesh C. and Harris Jr., Frederick C.},
  booktitle = {IEEE International Conference on Acoustics, Speech and Signal Processing (ICASSP)},
  pages = {1752--1756},
  year = {2008},
  url = {https://doi.org/10.1109/ICASSP.2008.4517969},
}

@inproceedings{du2012maxima,
  title = {Blind Robust Watermarking Mechanism Based on Maxima Curvature of 3D Motion Data},
  author = {Du, Ling and Cao, Xiaochun and Zhang, Muhua and Fu, Huazhu},
  booktitle = {Information Hiding (IH), Lecture Notes in Computer Science vol. 7692, Springer},
  pages = {110--124},
  year = {2012},
  url = {https://doi.org/10.1007/978-3-642-36373-3_8},
}

@inproceedings{zhu2018hidden,
  title = {HiDDeN: Hiding Data With Deep Networks},
  author = {Zhu, Jiren and Kaplan, Russell and Johnson, Justin and Fei-Fei, Li},
  booktitle = {European Conference on Computer Vision (ECCV)},
  year = {2018},
  eprint = {1807.09937},
  archivePrefix = {arXiv},
  url = {https://arxiv.org/abs/1807.09937},
}

@inproceedings{tancik2020stegastamp,
  title = {StegaStamp: Invisible Hyperlinks in Physical Photographs},
  author = {Tancik, Matthew and Mildenhall, Ben and Ng, Ren},
  booktitle = {IEEE/CVF Conference on Computer Vision and Pattern Recognition (CVPR)},
  year = {2020},
  eprint = {1904.05343},
  archivePrefix = {arXiv},
  url = {https://arxiv.org/abs/1904.05343},
}

@inproceedings{fernandez2023stablesignature,
  title = {The Stable Signature: Rooting Watermarks in Latent Diffusion Models},
  author = {Fernandez, Pierre and Couairon, Guillaume and Jégou, Hervé and Douze, Matthijs and Furon, Teddy},
  booktitle = {IEEE/CVF International Conference on Computer Vision (ICCV)},
  year = {2023},
  eprint = {2303.15435},
  archivePrefix = {arXiv},
  url = {https://arxiv.org/abs/2303.15435},
}

@inproceedings{wen2023treering,
  title = {Tree-Ring Watermarks: Fingerprints for Diffusion Images that are Invisible and Robust},
  author = {Wen, Yuxin and Kirchenbauer, John and Geiping, Jonas and Goldstein, Tom},
  booktitle = {Advances in Neural Information Processing Systems (NeurIPS)},
  year = {2023},
  eprint = {2305.20030},
  archivePrefix = {arXiv},
  url = {https://arxiv.org/abs/2305.20030},
}

@inproceedings{yang2024gaussianshading,
  title = {Gaussian Shading: Provable Performance-Lossless Image Watermarking for Diffusion Models},
  author = {Yang, Zijin and Zeng, Kai and Chen, Kejiang and Fang, Han and Zhang, Weiming and Yu, Nenghai},
  booktitle = {IEEE/CVF Conference on Computer Vision and Pattern Recognition (CVPR)},
  year = {2024},
  eprint = {2404.04956},
  archivePrefix = {arXiv},
  url = {https://arxiv.org/abs/2404.04956},
}

@inproceedings{huang2024robin,
  title = {ROBIN: Robust and Invisible Watermarks for Diffusion Models with Adversarial Optimization},
  author = {Huang, Huayang and Wu, Yu and Wang, Qian},
  booktitle = {Advances in Neural Information Processing Systems (NeurIPS)},
  year = {2024},
  eprint = {2411.03862},
  archivePrefix = {arXiv},
  url = {https://arxiv.org/abs/2411.03862},
}

@inproceedings{arabi2025hiddennoise,
  title = {Hidden in the Noise: Two-Stage Robust Watermarking for Images},
  author = {Arabi, Kasra and Feuer, Benjamin and Witter, R. Teal and Hegde, Chinmay and Cohen, Niv},
  booktitle = {International Conference on Learning Representations (ICLR)},
  year = {2025},
  eprint = {2412.04653},
  archivePrefix = {arXiv},
  url = {https://arxiv.org/abs/2412.04653},
}

@article{li2025stdmdiffusion,
  title = {Robust watermarking for diffusion models based on STDM and latent space fine-tuning},
  author = {Li, Lei and Zhang, Xinpeng and Feng, Guorui and Wang, Zichi and Wu, Dan and Wu, Hanzhou},
  journal = {Journal of Information Security and Applications},
  year = {2025},
  url = {https://doi.org/10.1016/j.jisa.2025.104167},
}

@inproceedings{kirchenbauer2023watermark,
  title = {A Watermark for Large Language Models},
  author = {Kirchenbauer, John and Geiping, Jonas and Wen, Yuxin and Katz, Jonathan and Miers, Ian and Goldstein, Tom},
  booktitle = {International Conference on Machine Learning (ICML)},
  year = {2023},
  eprint = {2301.10226},
  archivePrefix = {arXiv},
  url = {https://arxiv.org/abs/2301.10226},
}

@inproceedings{sanroman2024audioseal,
  title = {Proactive Detection of Voice Cloning with Localized Watermarking},
  author = {San Roman, Robin and Fernandez, Pierre and Défossez, Alexandre and Furon, Teddy and Tran, Tuan and Elsahar, Hady},
  booktitle = {International Conference on Machine Learning (ICML)},
  year = {2024},
  eprint = {2401.17264},
  archivePrefix = {arXiv},
  url = {https://arxiv.org/abs/2401.17264},
}

@misc{jang2024lvmark,
  title = {LVMark: Robust Watermark for Latent Video Diffusion Models},
  author = {Jang, Youngdong and Jang, MinHyuk and Lee, JaeHyeok and Yang, Feng and Oh, Gyeongrok and Jeong, Jongheon and Kim, Sangpil},
  howpublished = {arXiv preprint arXiv:2412.09122},
  year = {2024},
  eprint = {2412.09122},
  archivePrefix = {arXiv},
  url = {https://arxiv.org/abs/2412.09122},
}

@misc{li2025gaussianseal,
  title = {GaussianSeal: Rooting Adaptive Watermarks for 3D Gaussian Generation Model},
  author = {Li, Runyi and Zhang, Xuanyu and Tong, Chuhan and Xu, Zhipei and Zhang, Jian},
  howpublished = {arXiv preprint arXiv:2503.00531},
  year = {2025},
  eprint = {2503.00531},
  archivePrefix = {arXiv},
  url = {https://arxiv.org/abs/2503.00531},
}

@inproceedings{chen2025guardsplat,
  title = {GuardSplat: Efficient and Robust Watermarking for 3D Gaussian Splatting},
  author = {Chen, Zixuan and Wang, Guangcong and Zhu, Jiahao and Lai, Jianhuang and Xie, Xiaohua},
  booktitle = {IEEE/CVF Conference on Computer Vision and Pattern Recognition (CVPR)},
  year = {2025},
  eprint = {2411.19895},
  archivePrefix = {arXiv},
  url = {https://arxiv.org/abs/2411.19895},
}

@inproceedings{muller2025blackbox,
  title = {Black-Box Forgery Attacks on Semantic Watermarks for Diffusion Models},
  author = {Müller, Andreas and Lukovnikov, Denis and Thietke, Jonas and Fischer, Asja and Quiring, Erwin},
  booktitle = {IEEE/CVF Conference on Computer Vision and Pattern Recognition (CVPR)},
  pages = {20937--20946},
  year = {2025},
  eprint = {2412.03283},
  archivePrefix = {arXiv},
  url = {https://arxiv.org/abs/2412.03283},
}

@inproceedings{guo2022humanml3d,
  title = {Generating Diverse and Natural 3D Human Motions from Text},
  author = {Guo, Chuan and Zou, Shihao and Zuo, Xinxin and Wang, Sen and Ji, Wei and Li, Xingyu and Cheng, Li},
  booktitle = {IEEE/CVF Conference on Computer Vision and Pattern Recognition (CVPR)},
  pages = {5142--5151},
  year = {2022},
  url = {https://doi.org/10.1109/CVPR52688.2022.00509},
}

@inproceedings{mahmood2019amass,
  title = {AMASS: Archive of Motion Capture as Surface Shapes},
  author = {Mahmood, Naureen and Ghorbani, Nima and Troje, Nikolaus F. and Pons-Moll, Gerard and Black, Michael J.},
  booktitle = {IEEE/CVF International Conference on Computer Vision (ICCV)},
  year = {2019},
  eprint = {1904.03278},
  archivePrefix = {arXiv},
  url = {https://arxiv.org/abs/1904.03278},
}

@article{mason2022style100,
  title = {Real-Time Style Modelling of Human Locomotion via Feature-Wise Transformations and Local Motion Phases},
  author = {Mason, Ian and Starke, Sebastian and Komura, Taku},
  journal = {Proceedings of the ACM on Computer Graphics and Interactive Techniques (SIGGRAPH / I3D)},
  year = {2022},
  eprint = {2201.04439},
  archivePrefix = {arXiv},
  url = {https://arxiv.org/abs/2201.04439},
}

@article{harvey2020lafan1,
  title = {Robust Motion In-betweening},
  author = {Harvey, F{\'e}lix G. and Yurick, Mike and Nowrouzezahrai, Derek and Pal, Christopher J.},
  journal = {ACM Transactions on Graphics (SIGGRAPH)},
  volume = {39},
  number = {4},
  articleno = {60},
  year = {2020},
}

@inproceedings{tevet2023mdm,
  title = {Human Motion Diffusion Model},
  author = {Tevet, Guy and Raab, Sigal and Gordon, Brian and Shafir, Yonatan and Cohen-Or, Daniel and Bermano, Amit H.},
  booktitle = {International Conference on Learning Representations (ICLR)},
  year = {2023},
  eprint = {2209.14916},
  archivePrefix = {arXiv},
  url = {https://arxiv.org/abs/2209.14916},
}

@inproceedings{guo2024momask,
  title = {MoMask: Generative Masked Modeling of 3D Human Motions},
  author = {Guo, Chuan and Mu, Yuxuan and Javed, Muhammad Gohar and Wang, Sen and Cheng, Li},
  booktitle = {IEEE/CVF Conference on Computer Vision and Pattern Recognition (CVPR)},
  year = {2024},
  eprint = {2312.00063},
  archivePrefix = {arXiv},
  url = {https://arxiv.org/abs/2312.00063},
}

@inproceedings{zhang2023t2mgpt,
  title = {T2M-GPT: Generating Human Motion from Textual Descriptions with Discrete Representations},
  author = {Zhang, Jianrong and Zhang, Yangsong and Cun, Xiaodong and Huang, Shaoli and Zhang, Yong and Zhao, Hongwei and Lu, Hongtao and Shen, Xi},
  booktitle = {IEEE/CVF Conference on Computer Vision and Pattern Recognition (CVPR)},
  year = {2023},
  eprint = {2301.06052},
  archivePrefix = {arXiv},
  url = {https://arxiv.org/abs/2301.06052},
}

@inproceedings{pavllo2019videopose3d,
  title = {3D Human Pose Estimation in Video with Temporal Convolutions and Semi-Supervised Training},
  author = {Pavllo, Dario and Feichtenhofer, Christoph and Grangier, David and Auli, Michael},
  booktitle = {IEEE/CVF Conference on Computer Vision and Pattern Recognition (CVPR)},
  year = {2019},
  eprint = {1811.11742},
  archivePrefix = {arXiv},
  url = {https://arxiv.org/abs/1811.11742},
}

@inproceedings{shin2024wham,
  title = {WHAM: Reconstructing World-grounded Humans with Accurate 3D Motion},
  author = {Shin, Soyong and Kim, Juyong and Halilaj, Eni and Black, Michael J.},
  booktitle = {IEEE/CVF Conference on Computer Vision and Pattern Recognition (CVPR)},
  year = {2024},
  eprint = {2312.07531},
  archivePrefix = {arXiv},
  url = {https://arxiv.org/abs/2312.07531},
}

@inproceedings{li2021hybrik,
  title = {HybrIK: A Hybrid Analytical-Neural Inverse Kinematics Solution for 3D Human Pose and Shape Estimation},
  author = {Li, Jiefeng and Xu, Chao and Chen, Zhicun and Bian, Siyuan and Yang, Lixin and Lu, Cewu},
  booktitle = {IEEE/CVF Conference on Computer Vision and Pattern Recognition (CVPR)},
  year = {2021},
  eprint = {2011.14672},
  archivePrefix = {arXiv},
  url = {https://arxiv.org/abs/2011.14672},
}

@article{loper2015smpl,
  title = {{SMPL}: A Skinned Multi-Person Linear Model},
  author = {Loper, Matthew and Mahmood, Naureen and Romero, Javier and Pons-Moll, Gerard and Black, Michael J.},
  journal = {ACM Transactions on Graphics (Proc. SIGGRAPH Asia)},
  volume = {34},
  number = {6},
  pages = {248:1--248:16},
  year = {2015},
}

@misc{jocher2023yolov8,
  title = {Ultralytics {YOLOv8}},
  author = {Jocher, Glenn and Chaurasia, Ayush and Qiu, Jing},
  year = {2023},
  howpublished = {\url{https://github.com/ultralytics/ultralytics}},
}

@inproceedings{christ2024undetectable,
  title = {Undetectable Watermarks for Language Models},
  author = {Christ, Miranda and Gunn, Sam and Zamir, Or},
  booktitle = {Proceedings of the 37th Conference on Learning Theory (COLT)},
  series = {Proceedings of Machine Learning Research},
  volume = {247},
  pages = {1125--1139},
  year = {2024},
  url = {https://proceedings.mlr.press/v247/christ24a.html},
}

@article{fairoze2025publicly,
  title = {Publicly-Detectable Watermarking for Language Models},
  author = {Fairoze, Jaiden and Garg, Sanjam and Jha, Somesh and Mahloujifar, Saeed and Mahmoody, Mohammad and Wang, Mingyuan},
  journal = {{IACR} Communications in Cryptology},
  volume = {1},
  number = {4},
  year = {2025},
  doi = {10.62056/ahmpdkp10},
}

@article{chaum1985security,
  title = {Security without Identification: Transaction Systems to Make Big Brother Obsolete},
  author = {Chaum, David},
  journal = {Communications of the ACM},
  volume = {28},
  number = {10},
  pages = {1030--1044},
  year = {1985},
  doi = {10.1145/4372.4373},
}

@inproceedings{camenisch2001credentials,
  title = {An Efficient System for Non-transferable Anonymous Credentials with Optional Anonymity Revocation},
  author = {Camenisch, Jan and Lysyanskaya, Anna},
  booktitle = {Advances in Cryptology, EUROCRYPT 2001},
  series = {LNCS},
  volume = {2045},
  pages = {93--118},
  year = {2001},
}

@inproceedings{pfeuffer2019behavioural,
  title = {Behavioural Biometrics in {VR}: Identifying People from Body Motion and Relations in Virtual Reality},
  author = {Pfeuffer, Ken and Geiger, Matthias J. and Prange, Sarah and Mecke, Lukas and Buschek, Daniel and Alt, Florian},
  booktitle = {Proceedings of the 2019 CHI Conference on Human Factors in Computing Systems (CHI)},
  year = {2019},
  doi = {10.1145/3290605.3300340},
}

@article{cox1997spread,
  title = {Secure Spread Spectrum Watermarking for Multimedia},
  author = {Cox, Ingemar J. and Kilian, Joe and Leighton, F. Thomson and Shamoon, Talal},
  journal = {IEEE Transactions on Image Processing},
  volume = {6},
  number = {12},
  pages = {1673--1687},
  year = {1997},
  doi = {10.1109/83.650120},
}

@article{wong2001secret,
  title = {Secret and Public Key Image Watermarking Schemes for Image Authentication and Ownership Verification},
  author = {Wong, Ping Wah and Memon, Nasir},
  journal = {IEEE Transactions on Image Processing},
  volume = {10},
  number = {10},
  pages = {1593--1601},
  year = {2001},
  doi = {10.1109/83.951543},
}

@inproceedings{miller2020vrauth,
  title = {Within-System and Cross-System Behavior-Based Biometric Authentication in Virtual Reality},
  author = {Miller, Robert and Banerjee, Natasha Kholgade and Banerjee, Sean},
  booktitle = {IEEE Conference on Virtual Reality and 3D User Interfaces Abstracts and Workshops (VRW)},
  pages = {311--316},
  year = {2020},
}

\appendix
\section{Related Work}\label{sec:related}
\noindent\textbf{VR motion biometrics and anonymization.}
Nair et al.\ established large-scale identification from head and hand
motion~\cite{nair2023unique}, attribute inference~\cite{nair2023inferring}, and
the BOXRR-23 dataset~\cite{nair2024boxrr}; Schach et al.\ measure how far it
carries between extended-reality (XR) applications and find that it generalizes
weakly~\cite{schach2025crossxr}. Deep Motion
Masking~\cite{nair2024deepmasking} is the state-of-the-art anonymizer, and
AvatarHunter~\cite{meng2024avatarhunter} demonstrates de-anonymization from
recorded avatar gait. These works define the threat and the operating regime that
MoSign augments with authentication.

\smallskip\noindent\textbf{Authenticating without identifying.}
The goal of proving something about a party without revealing who they are originates with
anonymous credentials~\cite{chaum1985security,camenisch2001credentials}, which authenticate the
holder of a credential at the moment a session is established. That guarantee is bound to a
handshake, not to what subsequently travels over the channel, so it leaves open whether the motion
arriving now is the motion the credential holder produced; MoSign carries the proof in the stream
itself, which is what lets it reject stale replay and splicing. At the other
end of the design space, VR behavioural biometrics authenticate directly from
motion~\cite{pfeuffer2019behavioural,miller2020vrauth}, but do so by recognizing the person, which
is the capability the anonymization requirement removes.

\smallskip\noindent\textbf{Classical motion and animation watermarking.}
Robust and fragile watermarking of motion-capture streams and skinned animation
embeds marks via skeleton dynamics, key-framed animation, skin weights, or
spatio-temporal curvature~\cite{agarwal2007tamper,li2007progressive,
motwani2008skinning,du2012maxima}. They are post-hoc, hand-crafted, and not coupled to
a generator; their goals are copyright and tamper detection, the latter a form of content
authentication, as in the fragile and semi-fragile image schemes that verify media has not
been altered since signing~\cite{wong2001secret}, rather than of the entity authentication we
target: an artifact is vouched for after the fact, whereas we vouch for a live party under a
per-epoch challenge, which makes freshness and replay rather than tamper localization the
binding constraints. Our Post-hoc spread-spectrum
baseline is a construction in this style,
keyed so that it too renews per epoch, which isolates what the in-generation route buys.

\smallskip\noindent\textbf{In-generation and learned watermarking.}
HiDDeN~\cite{zhu2018hidden} introduced differentiable encoder/decoder training
with a noise layer, which underlies our recapture-robust training;
StegaStamp~\cite{tancik2020stegastamp} survives a print-then-photograph loop, the
physical analogue of recapture. For diffusion, Stable
Signature~\cite{fernandez2023stablesignature} roots a mark in the latent decoder;
Tree-Ring~\cite{wen2023treering} and Hidden-in-the-Noise~\cite{arabi2025hiddennoise}
mark the initial noise; ROBIN~\cite{huang2024robin} and STDM-based
schemes~\cite{li2025stdmdiffusion} push the robustness/imperceptibility frontier.
\emph{Gaussian Shading}~\cite{yang2024gaussianshading} is provably
distribution-preserving and is the direct basis of our embedding, which we extend with a
per-epoch keyed renewal so that undetectability holds against an observer who collects many
samples. Beyond images,
LVMark~\cite{jang2024lvmark} decodes video watermarks in a low-frequency
3D-wavelet band (a temporal-robustness idea we reuse), and
GaussianSeal~\cite{li2025gaussianseal} and GuardSplat~\cite{chen2025guardsplat}
watermark 3D Gaussian generation. These establish ``first generative watermark for
a new modality'' as a recognized contribution, which we provide for skeletal
motion while additionally elevating the goal from provenance to authentication.

\smallskip\noindent\textbf{Statistical detection and attacks.}
The LLM $z$-test framework~\cite{kirchenbauer2023watermark} grounds our
hypothesis-testing decision; AudioSeal~\cite{sanroman2024audioseal} performs
localized per-frame detection, which fits streaming VR. Recent black-box
\emph{forgery} of semantic watermarks~\cite{muller2025blackbox} motivates our
unforgeability analysis and adaptive-adversary evaluation.
Our verification is symmetric: the verifier holds $\key$. Publicly-detectable constructions
for language models obtain asymmetric verification from digital signatures~\cite{fairoze2025publicly},
at a payload cost that a $63$-bit motion epoch cannot currently absorb, a limitation we return to at the end.

\smallskip\noindent\textbf{Motion generation, data, and 3D pose estimation.}
MoSign's generator is a motion VAE, but the embedding primitive needs only a Gaussian latent, so it also
runs inside a frozen, externally trained Motion Diffusion Model (MDM)~\cite{tevet2023mdm}. Discrete
token-logit backbones such as T2M-GPT~\cite{zhang2023t2mgpt} and MoMask~\cite{guo2024momask} carry the
message in a non-Gaussian latent and are left to future work. The datasets we use are
HumanML3D~\cite{guo2022humanml3d} (built on AMASS~\cite{mahmood2019amass}),
LAFAN1~\cite{harvey2020lafan1}, 100STYLE~\cite{mason2022style100}, and
BOXRR-23~\cite{nair2024boxrr}. Monocular 3D pose estimation supplies the recapture
channel~\cite{pavllo2019videopose3d}, with WHAM~\cite{shin2024wham} and HybrIK~\cite{li2021hybrik} as
references for stronger estimators than the one we attack with.


\section{Generality of the embedding primitive}\label{app:general}
The body evaluates the deployed configuration: one generator, one dataset, one payload length.
This appendix reports the four experiments that ask how far the construction reaches beyond it.
None of them is required for the claims of the paper; each says something about where the
primitive would still work.

\smallskip\noindent\textbf{Zero-shot transfer to unseen motion distributions.}
The verifier's read-out also transfers \emph{zero-shot}. Applying the frozen HumanML3D
watermark to two unseen datasets, LAFAN1~\cite{harvey2020lafan1} and
100STYLE~\cite{mason2022style100} ($170$ and $900$ windows), the clean codeword reads at TPR $1.0$ (codeword accuracy
$0.96$ on both), every single benign channel stays intact, and only crop, dropout and their
seven-attack composition degrade ($0.80$ to $0.92$), so the extractor is not overfit to
HumanML3D content. The generator is
dataset-specific, however: on these unseen motions the HumanML3D-trained VAE mispredicts the global root
trajectory (about $2$\,m of global drift, though the local pose degrades only modestly, root-relative mean
per-joint position error (MPJPE)
$0.2$ versus $0.04$ in-distribution), a generic root-velocity-integration limitation that the watermark does
\emph{not} cause, since anonymized and watermarked motion drift within $0.02$\,m of each other. A deployment on a new distribution
thus retrains the generator and its jointly-trained extractor, while the cryptographic codec and the
constructional undetectability and $(\efa,Q)$ guarantees carry over.

\paragraph{The embedding primitive ports to an external generator.}
The embedding is also not tied to our VAE. We instantiate it in a \emph{frozen}, externally trained
MDM~\cite{tevet2023mdm} (Table~\ref{tab:backbones}): the keyed codeword shades MDM's initial diffusion noise
$x_T$ exactly as it shades $\zp$, the frozen sampler runs to a motion, and the codeword is read back by DDIM
inversion. A diffusion backbone is invertible, so the read-out is the generator's own inverse and
\emph{nothing is trained}: across $N{=}128$ generations the clean codeword reads at codeword accuracy $0.995$
(TPR $1.0$), the written $x_T$ stays $\Norm(0,I)$ (latent WM-IND AUC $0.52$), and a wrong-nonce read is chance
($0.50$). This is the primitive in a text-conditioned, pure-generation backbone; the motion-conditioned
reconstruction we deploy keeps the VAE, whose separate content ($\zm$) and style ($\zp$) latents let the
watermark ride in $\zp$ without disturbing the user's motion.

\begin{table}[tbp]
\centering
\footnotesize
\setlength{\tabcolsep}{4pt}
\caption{The embedding primitive across generator backbones. VAE is our motion-conditioned model (deployed;
learned extractor); MDM is a frozen external text-to-motion diffusion model read by denoising diffusion implicit model (DDIM) inversion (no trained
reader). WM-IND AUC and replay are inherited from the shared codec.}
\label{tab:backbones}
\setlength{\tabcolsep}{3pt}
\begin{tabular}{llcc}
\toprule
Backbone & Read-out & Codeword & WM-IND\\
\midrule
VAE (deployed) & extractor & $0.98$  & $0.51$\\
MDM (frozen)   & DDIM inv. & $0.995$ & $0.52$\\
\bottomrule
\end{tabular}
\end{table}

\smallskip\noindent\textbf{Payload length.}
The operating point embeds a $63$-bit codeword per epoch. To map the
capacity-robustness tradeoff we trained otherwise-identical models at codeword
lengths $\nb\in\{31,63,127,255\}$ (replication $\dq/\nb$ from $8$ down
to $1$) and measured the per-read codeword accuracy
(Figure~\ref{fig:payload}). At high redundancy the empirical reliability tracks
the idealized repetition-coding bound
$p_1(\nb)=\Phi(\sqrt{\dq/\nb}\,\Phi^{-1}(q))$ (dashed), where $q$ is the
per-dimension sign accuracy calibrated at $\nb{=}63$, but past
$\nb=63$ it falls well below it and plateaus near the raw per-dimension accuracy
($0.73$ clean at $\nb=127$ versus the bound's $0.92$): once redundancy runs out the
embedding cannot hold per-dimension quality up, so the practical capacity is
bounded rather than following the idealized curve. The default $63$-bit codeword (a
$\mathrm{BCH}(63,30,t{=}6)$ message) sits at the knee, in the regime where every channel's
per-window accuracy clears the one-epoch threshold of the sequential test.

\begin{figure}[tbp]\centering
\includegraphics[width=\linewidth]{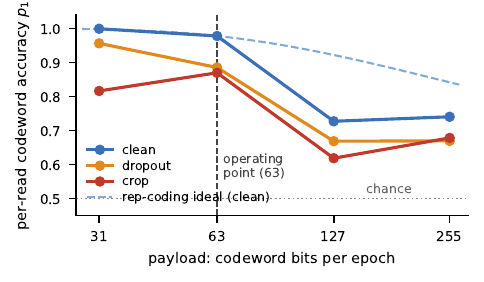}
\caption{Capacity-robustness tradeoff, measured. Per-read codeword accuracy $p_1$
for models trained at codeword length $\nb\in\{31,63,127,255\}$ (each an
otherwise-identical model, same training recipe), on the clean, dropout, and crop channels.
The dashed curve is the idealized repetition-coding bound calibrated to the clean
$\nb=63$ point. The default $63$-bit codeword is marked.}
\label{fig:payload}
\end{figure}

\smallskip\noindent\textbf{Epoch length.}
The shipped $196$-frame epoch makes the prover buffer $9.8$\,s before it can emit
(Section~\ref{sec:disc}). Sweeping the epoch over $\{40,64,100,196\}$ frames and re-reading every
channel says that buffer is close to free. The $40$-frame epoch buffers $2.0$\,s and costs at most
$0.01$ of codeword accuracy anywhere: clean $0.972$ against $0.976$, recapture $0.958$ against
$0.964$, dropout $0.867$ against $0.877$, and composition with the anonymizer $0.940$ against $0.949$.
On cropping, the long epoch's weakest benign channel, it is strictly better, $0.972$ against
$0.863$. A crop keeps a random sub-segment from a random offset, so it shifts the epoch phase as well
as shortening the clip, and at $196$ frames a clip is about one epoch with nothing to fall back on,
whereas $40$ frames give the verifier several epochs to align with (Figure~\ref{fig:streaming}); the sweep does
not separate the two. A deployment can therefore cut prover-side lookahead to $2$\,s at negligible
cost in recovery. That figure is the prover's buffer alone; the verifier's time-to-authenticate runs
on top of it. These are the sweep's own numbers, not
Table~\ref{tab:robust}'s configuration.

\section{Removal attacks in detail}\label{app:removal}
The body states what removal buys an adversary and why it is bounded. This appendix gives the
attack-by-attack measurements behind that statement: the distortion each attack needs to drive the
verifier to reject, what the per-epoch policy does when part of a stream is destroyed, and an
adaptive laundering attack trained against the recapture channel.

\begin{figure}[tbp]\centering
\includegraphics[width=\linewidth]{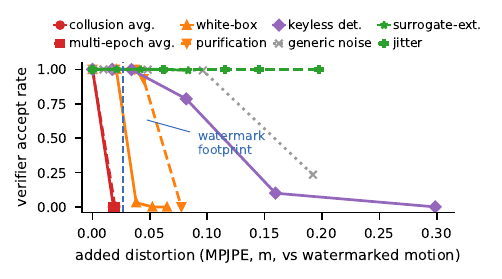}
\caption{The removal attacks on one axis: verifier accept rate versus added distortion,
grouped by adversary class. Blue dashed is the watermark's own footprint.}
\label{fig:removal}
\end{figure}

\begin{table}[tbp]\centering\footnotesize\caption{Removal attacks (companion to Figure~\ref{fig:removal}). Each row is the smallest added MPJPE, on
that attack's own strength grid, at which the verifier accept rate first reaches $0$; the watermark's
own footprint is $0.027$. Only
cross-sample averaging (bold) strips the mark
below that footprint, and only when the same motion is available under $K$ keys or across $K$ epochs.
Three never reach zero: the surrogate extractor leaves $0.99$ at $0.083$, rate/phase jitter $1.00$ at
$0.197$, and generic noise $0.23$ at $0.192$.
\textsuperscript{$\ddagger$}Re-generation is quoted root-relative; its raw $0.711$ is dominated by
global root drift.}\label{tab:adaptive}
\setlength{\tabcolsep}{3pt}
\begin{tabular}{lcc}
\toprule
Removal attack & needs & auth.\ $\to 0$ at\\
               &       & (added MPJPE, m)\\
\midrule
collusion avg.\ & $K$ keys & $\mathbf{0.018}$\\
multi-epoch avg.\ & $K$ epochs & $\mathbf{0.020}$\\
white-box PGD       & extractor wts & $0.052$\\
re-generation       & generative AE & $0.065$\textsuperscript{$\ddagger$}\\
purification AE     & generic AE    & $0.077$\\
keyless detector    & trains det.\  & $0.299$\\
generic noise       & none          & never (to $0.23$)\\
surrogate extr.\ & trains surr.\ & never (to $0.99$)\\
rate/phase jitter   & none          & never (to $1.00$)\\
\bottomrule
\end{tabular}

\end{table}

\paragraph{The removal attacks.}
Figure~\ref{fig:removal} places every removal attack on one accept-versus-distortion
axis, and Table~\ref{tab:adaptive} lists the distortion each needs to drive the verifier
to reject; the axis is the injected MPJPE, measured against the observed watermarked
motion. The attacks fall into four classes. \emph{(i)~Cross-sample averaging}
(multi-epoch or collusion) needs the least distortion, cancelling the keyed carrier at
${\sim}0.02$\,m added MPJPE, below the watermark's own $0.027$\,m footprint, by averaging
out a code that differs across epochs or keys rather than fighting the extractor. What it
needs instead is the same motion recorded $K$ times: either one performance re-signed at
$K$ epochs, which a live session does not produce because the next epoch carries different
motion, or $K$ signers performing one trajectory under $K$ real keys inside a single session,
sharing the verifier's nonce. Both are protocol-level
preconditions rather than signal processing.
\emph{(ii)~Single-clip gradient removal} needs real distortion: white-box projected gradient descent (PGD), granted
the extractor weights as in the forgery benchmark, collapses authentication only near the
footprint, and a purification autoencoder needs $0.077$\,m. \emph{(iii)~A keyless presence
detector}, a classifier trained on signed-versus-anonymized motion and run as an
$\ell_\infty$ gradient attack, does gain a genuine advantage over generic noise,
driving the accept rate to $0.10$ by ${\sim}0.16$\,m and to zero by $0.30$\,m. The
detector is a poor \emph{distinguisher} (its own WM-IND AUC is $0.51$, interval
$[0.49,0.54]$) but a useful \emph{gradient}: both branches put energy in the same
high-frequency band, and descending that gradient degrades the band without separating
the branches or reading the keyed codeword. Undetectability constrains what an observer
can \emph{learn}, not what a gradient can \emph{destroy}, which is the
undetectability/robustness tension on the active axis. \emph{(iv)~But
reading the code fails}: a surrogate-\emph{extractor} transfer attack and rate/phase
jitter never suppress the verifier ($0.99$ at $0.08$\,m, $1.0$ at $0.20$\,m), because the
codeword is keyed and a keyless surrogate cannot predict it. \emph{(v)~Re-generation}, the
canonical latent-watermark attack, re-encodes the observed motion through the public VAE and
re-samples $\zp$ (the inverse of the embedding). It too strips the mark to chance, but its cost
is global, not local: MoSign's VAE is not cycle-consistent on its own generations, so the
round-trip keeps the local pose (root-relative MPJPE $0.07$) while drifting the global trajectory
$0.71$, a locally faithful but globally displaced laundering kept off the local-distortion axis of
Figure~\ref{fig:removal}. No removal is impersonation: forgery fails, so stripping an observed signature only denies a single
observation. A splicing adversary that assembles candidate streams from observed epochs
gains nothing either: its per-stream false-accept is $0$ over the measured trials, and
the best-of-$Q$ curve extrapolated from that rate stays at $0$ out to $Q{=}10^4$.
Splicing is the harder target of the two, since a spliced stream carries no epoch's
codeword intact, whereas the query adversary perturbs a genuine one.

\begin{figure}[tbp]\centering
\includegraphics[width=\linewidth]{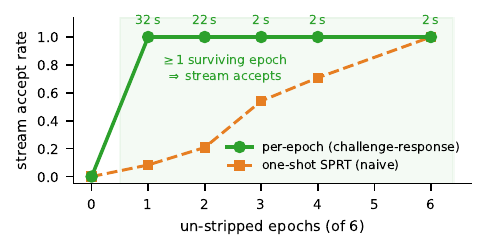}
\caption{Per-epoch authentication (green) versus a naive one-shot SPRT (orange), as a function
of how many of a six-epoch stream's epochs survive removal ($24$ streams, worst case: the rest
\emph{fully} stripped). Green labels are the time-to-authenticate.}
\label{fig:streaming}
\end{figure}

\paragraph{Per-epoch authentication and what removal buys.}
Because MoSign authenticates per epoch (the challenge-response renews each epoch), a
stream authenticates if any single epoch is left un-corrupted. Figure~\ref{fig:streaming}
strips a fraction of a six-epoch stream's epochs \emph{entirely} and shows the per-epoch
policy accepts every stream whenever ${\ge}1$ epoch survives, in any position, failing
only when all six are stripped, whereas a naive one-shot test over the whole stream degrades
monotonically ($1.00\!\to\!0.08$ as survivors fall from six to one). The cost is
graceful: the time-to-authenticate rises from $2$\,s to ${\sim}32$\,s as clean epochs grow scarce,
and the stream still authenticates. Corrupting \emph{every} epoch of a live stream with no gap requires a
continuous in-line presence on the motion channel; that adversary can already deny service by
blocking the stream, though laundering is quieter than blocking. The load-bearing
property is the other one: removal is not impersonation. A stripped stream authenticates as no one,
because forging a valid mark needs the key, so the strongest outcome
any removal reaches is denial of service. Figure~\ref{fig:streaming} strips epochs rather than
corrupting all of them at once, so it bounds the per-epoch policy's behaviour under partial
contamination rather than under averaging.

\begin{figure}[tbp]\centering
\includegraphics[width=\linewidth]{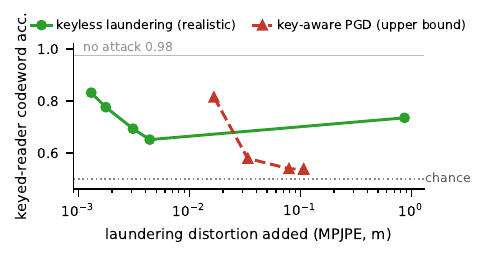}
\caption{Adaptive in-channel laundering of the recapture path: keyed-reader codeword accuracy
versus laundering distortion, for a realistic keyless laundering net (green) and a key-aware
white-box PGD upper bound (red).}
\label{fig:g9}
\end{figure}

\paragraph{Adaptive laundering of the recapture channel.}
The receiver-side keyed reader (C4) invites the recapture analogue of the removal benchmark: a
man-in-the-middle on the \emph{video} channel could launder the honest avatar to strip the key-holder's
read-out and reject a legitimate user (availability, not impersonation). We train a laundering net on the
avatar, push it through the differentiable recapture channel and the frozen keyed reader, and minimize the
reader's recovery under a distortion penalty (Figure~\ref{fig:g9}). A realistic keyless attacker degrades the
key-holder's recovery from $0.98$ to $0.65$ at $0.004$\,m added MPJPE, a sixth of the
watermark's own footprint, and never reaches chance anywhere in the sweep: that
$0.65$ is the attacker's best point, and at $0.87$\,m, thirty times that footprint, the key holder
still reads $0.74$.
Even an over-powered key-aware white-box PGD, granted the per-session key, floors at $0.54$ and never reaches
chance. As with the motion-domain removals, laundering forges nothing and per-epoch renewal
re-authenticates on the next clean epoch, so sustained denial demands continuous in-line channel
control.

\end{document}